\documentclass[11pt]{amsart} 
\pdfoutput=1

\usepackage{amssymb}
\usepackage{amsthm}
\usepackage{amsfonts}
\usepackage{amsmath}
\usepackage{pmboxdraw}
\usepackage{verbatim} 
\usepackage{graphicx}
\usepackage{color}
\usepackage[colorlinks=true, citecolor=blue, filecolor=black, linkcolor=black, urlcolor=black]{hyperref}
\usepackage{cite}
\usepackage[normalem]{ulem}
\usepackage{subcaption}
\usepackage{bbm}
\usepackage{bm}
\usepackage{mathtools}
\usepackage{todonotes}
\usepackage{calrsfs}

\newcommand{\RN}[1]{%
	\textup{\uppercase\expandafter{\romannumeral#1}}%
}

\def\pa{\partial}

\def\wt{\widetilde}

\def\AA{\textbf{A}}

\def\C{\mathbb{C}}

\def\R{\mathbb{R}}

\newcommand{\erfc}{\operatorname{erfc}}
\newcommand{\erf}{\operatorname{erf}}

\newcommand{\re}{\operatorname{Re}}
\newcommand{\im}{\operatorname{Im}}
\newcommand{\Tr}{\operatorname{Tr}}
\newcommand{\sgn}{\operatorname{sgn}}
\newcommand{\Var}{\operatorname{Var}}

\theoremstyle{plain}
\newtheorem*{thm*}{Theorem}
\newtheorem{thm}{Theorem}[section]
\newtheorem{lem}[thm]{Lemma}
\newtheorem{cor}[thm]{Corollary}
\newtheorem{prop}[thm]{Proposition}
\newtheorem*{prop*}{Proposition}
\newtheorem*{lem*}{Lemma}
\newtheorem{ex}[thm]{Example}
\newtheorem{rem}[thm]{Remark}

\theoremstyle{definition}
\newtheorem*{eg*}{Example}
\newtheorem*{egs*}{Examples}
\newtheorem*{def*}{Definition}
\newtheorem*{Q*}{Question}

\theoremstyle{remark}
\newtheorem*{rmk*}{Remark}
\newtheorem*{rmks*}{Remarks}

\numberwithin{equation}{section}

\begin{document}
\title[Large deviations and fluctuations of real elliptic random matrices]{
Large deviations and fluctuations of real eigenvalues of elliptic random matrices
}

\author{Sung-Soo Byun}
\address{Center for Mathematical Challenges, Korea Institute for Advanced Study, 85 Hoegiro, Dongdaemun-gu, Seoul 02455, Republic of Korea}
\email{sungsoobyun@kias.re.kr}

\author{Leslie Molag}
\address{Department of Mathematics, University of Sussex, Brighton, BN1 9RH, United Kingdom}
\email{L.D.Molag@sussex.ac.uk}

\author{Nick Simm}
\address{Department of Mathematics, University of Sussex, Brighton, BN1 9RH, United Kingdom}
\email{n.j.simm@sussex.ac.uk}

\date{\today}



\begin{abstract}
We study real eigenvalues of $N\times N$ real elliptic Ginibre matrices indexed by a non-Hermiticity parameter $0\leq \tau<1$, in both the strong and weak non-Hermiticity regime. Here $N$ is assumed to be an even number. In both regimes, we prove a central limit theorem for the number of real eigenvalues. We also find the asymptotic behaviour of the probability $p_{N,k}^{(\tau)}$ that exactly $k$ eigenvalues are real. In the strong non-Hermiticity regime, where $\tau$ is fixed, we find 
\begin{align*}
    \lim_{N\to\infty} \frac{1}{\sqrt{N}} \log p_{N,k_N}^{(\tau)} = -\sqrt\frac{1+\tau}{1-\tau} \frac{\zeta(3/2)}{\sqrt{2\pi}}
\end{align*}
for any sequence $(k_N)_N$ of even numbers such that $k_N = o(\frac{\sqrt N}{\log N})$ as $N\to\infty$, where $\zeta$ is the Riemann zeta function. 
In the weak non-Hermiticity regime, where $\tau=1-\frac{\alpha^2}{N}$, we obtain
\begin{align*}
    \lim_{N\to\infty} \frac{1}{N} \log p_{N,k_N}^{(\tau)} \leq  \frac{2}{\pi} \int_0^1 \log\left(1-e^{-\alpha^2 s^2}\right) \sqrt{1-s^2} \, ds
\end{align*}
for any sequence $(k_N)_N$ of even numbers such that $k_N=o(\frac{N}{\log N})$ as $n\to\infty$. This inequality is expected to be an equality. 
\end{abstract}

\maketitle

\section{Introduction and Main results}

In 1965, Ginibre introduced three random matrix models that are essentially the unconstrained versions of the GOE, GUE and GSE, i.e. all entries are i.i.d. Gaussians without the requirement of Hermiticity \cite{ginibre1965statistical}. Due to the lack of Hermiticity, the eigenvalues are not confined to the real line, and live on the full complex plane. These ensembles consist of $N\times N$ matrices $M$, with real ($\beta=1$), complex ($\beta=2$) or (real) quaternion ($\beta=4$) entries, that are distributed according to the probability measure
\begin{align*}
    \displaystyle\frac{1}{Z_N^\beta} e^{-\frac{1}{2} \beta \Tr(M^\dagger M)} dM_N^\beta,
\end{align*}
where $Z_N^\beta$ is a normalisation constant, and $dM_N^\beta$ is the standard Lebesgue measure on the corresponding spaces of matrices of real dimension $\beta N^2$.
Nowadays these models are called the real, complex and quaternion Ginibre ensembles (denoted as GinOE, GinUE and GinSE), and they are well-studied in the past half century. For a recent review on the Ginibre ensembles, we refer to the papers \cite{byun2022progress, byun2023progress}.  

In the present paper we focus on $\beta=1$. In fact, we consider a one-parameter deformation of the GinOE, called the real elliptic Ginibre ensemble (eGinOE). Likely inspired by Girko \cite{girko1986elliptic}, the eGinOE was introduced in 1988 by Sommers, Crisanti, Sompolinsky and Stein \cite{SoCrSoSt}. The eGinOE with parameter $-1<\tau<1$ consists of  $N\times N$ real matrices $M$, with centered Gaussian entries that satisfy the correlation structure
\begin{align*}
    \mathbb E M_{ij}^2 =\frac{1}{N}, \qquad \mathbb E M_{ij} M_{ji} = \frac{\tau}{N}, \qquad \mathbb E M_{ii}^2 = \frac{1+\tau}{N}, \qquad i,j=1,\ldots,N\text{ and }i\neq j.
\end{align*}
These are precisely the real random matrices $M$ that are distributed according to the probability measure
\begin{align*}
\displaystyle\frac{1}{Z_N^{(\tau)}} e^{-\frac{1}{2(1-\tau^2)}\Tr(M^\dagger M-\tau M^2)} dM_N, \qquad  dM_N = \prod_{i,j=1}^N dM_{ij},
\end{align*}
where $Z_N^{(\tau)}$ is a normalisation constant. For $\tau=0$, we obtain the GinOE. On the other hand, in the limit $\tau\uparrow 1$, it is known that the eGinOE approaches the GOE. In the limit $\tau\downarrow -1$, the matrix $M$ is real and anti-symmetric. This is equivalent (after multiplication by $i$) to what is known as the anti-symmetric GUE, see \cite{MR2742822,MR2663989} and references therein for further details about this ensemble.  
One can define the eGinOE equivalently as the ensemble consisting of $N\times N$ matrices
\begin{equation}\label{X GOE antiGOE}
M := \sqrt{\frac{1+\tau}{2}}\, S + \sqrt{\frac{1-\tau}{2}} \, A,
\end{equation}
where $S$ and $A$ are matrices picked from the GOE and its anti-symmetric version. 
Nowadays, most authors require the parameter $\tau$ to be in $(0,1)$ or $[0,1)$ in the definition of the eGinOE. 

The first occurence of a GinOE matrix in an application was in a paper by May \cite{may1972will} in 1972, who investigated complex ecological systems. More precisely, May
considered the stability of the solutions to
\begin{align*}
\vec x' = \left(-\mathbb I+\alpha G_N\right) \vec x,
\end{align*}
where $G_N$ is a GinOE matrix. This allows to investigate general systems of differential equations $\vec x' = f(\vec x)$ where $f$ is unknown. 
More general systems of differential equations associated with the eGinOE were further investigated by Fyodorov and Khoruzhenko \cite{MR3521630}. 
Over the years many other applications have been introduced, ranging from dynamics to random networks and cortical electric activity to quantum chromodynamics, see \cite{MR2185860} for references. 

For general fixed $\tau \in [0, 1)$ and $N \to \infty$, it is well known that the eigenvalues are uniformly distributed on the ellipse
\begin{equation} \label{Ellipse}
\mathcal E^{(\tau)}:=\Big \{ z\in\mathbb C: \Big( \frac{\re z}{1+\tau} \Big)^2+\Big( \frac{\im z}{1-\tau} \Big)^2 \le 1 \Big\}, 
\end{equation}
see e.g. \cite{girko1986elliptic,SoCrSoSt,MR3403996}.
We also refer to \cite{alt2022local} for the local law for elliptic random matrices. 

\begin{figure}[h!]
    \begin{subfigure}{0.48\textwidth}
       \begin{center}
            \includegraphics[width=\textwidth]{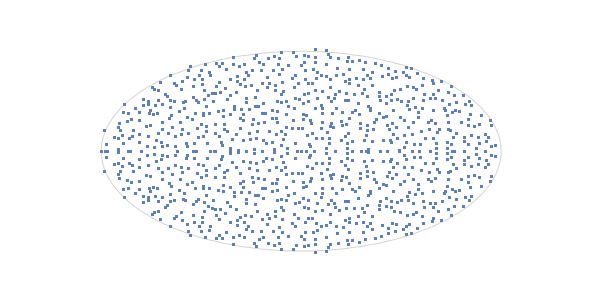}
        \end{center}
    	\subcaption{$N=1000$}
	\end{subfigure}
	\begin{subfigure}{0.48\textwidth}
		\begin{center}
		\includegraphics[width=\textwidth]{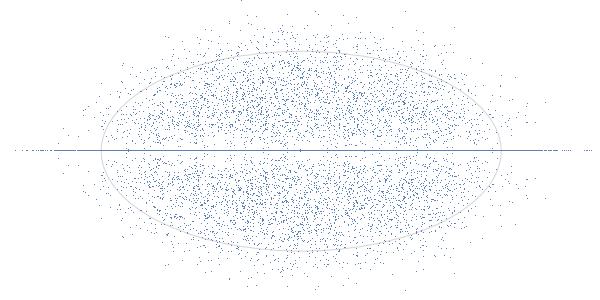}
		\end{center}
		\subcaption{ $N=10$ with $1000$ repetitions }
	\end{subfigure} \caption{Eigenvalues of the eGinOE.} \label{Fig_eGinOE}
\end{figure}

In his original paper, Ginibre only managed to derive the joint probability density function (JPDF) for the case that all $N$ eigenvalues are real. It took about a quarter century longer before the JPDF in the general case was determined \cite{MR1121461}.  When $N=2n$, and the number of real eigenvalues is $2k$, the corresponding JPDF is given by
\begin{align*}
C_n^{(\tau)} \binom{n}{k} 2^{n-k} \prod_{j=1}^k \sqrt{\omega^{(\tau)}(\lambda_j)} \prod_{\ell=1}^{n-k} \omega^{(\tau)}(z_\ell)
|\Delta\left(\lambda_1, \ldots, \lambda_k, z_1, \overline{z}_1, \ldots, z_{n-k}, \overline{z}_{n-k}\right)|,
\end{align*}
where $C_n^{(\tau)}$ is an explicit constant, $$\omega^{(\tau)}(z) = e^{-\frac{\re(z^2)}{1+\tau}} \erfc\Big(\sqrt\frac{2}{1-\tau^2} \im z\Big),$$ 
and $\Delta$ is a Vandermonde factor. We mention that the associated correlation functions exhibit a Pfaffian structure \cite{KanzieperAkemann07}. 

The difficulty in determining the JPDF for the GinOE and eGinOE, is due to the unique property, not seen in the complex and quaternion counterparts, that purely real eigenvalues occur with non-zero probability, see Figure~\ref{Fig_eGinOE}. 
In particular, the probabilities
\begin{align} \label{def:Pnktau}
    p_{N,k}^{(\tau)}, \qquad k = 0,1,\ldots, N
\end{align}
that a particular eigenvalue configuration of the eGinOE (and GinOE) has exactly $k$ real eigenvalues are non-zero, when $k$ has the same parity as $N$.

In the current paper, we shall focus on the case $N$ is even as the odd $N$ case requires separate treatment, see e.g. \cite{MR2485724,MR2341601}. We shall write
\begin{align}
    N = 2n, \qquad n=1,2,\ldots
\end{align}
henceforth.

\subsection{Main result for fluctuations of the number of real eigenvalues}

 Let $\mathcal N_N^{(\tau)}$ be the number of real eigenvalues of $M$.
For $\tau \in [0, 1)$ fixed, it was shown by Forrester and Nagao \cite{MR2430570} that the expected number of real eigenvalues is given by
\begin{equation}\label{EN tau<1}
	\mathbb E \mathcal N_N^{(\tau)}= \sqrt{\frac{1+\tau}{1-\tau}}\sqrt{\frac{2 N}{\pi}} (1+o(1)), \qquad ( N \to \infty).
\end{equation}
The formula \eqref{EN tau<1} was first proved by Edelman, Kostlan, and Shub for the GinOE case ($\tau=0$) \cite{MR1231689}.
In addition to the mean, the variance of the number of real eigenvalues satisfies the asymptotic behaviour
\begin{equation} \label{VN tau strong}
	\Var \mathcal N_N^{(\tau)} = (2-\sqrt{2})\sqrt{\frac{1+\tau}{1-\tau}}\sqrt{\frac{2 N}{\pi}}(1+o(1)),
 \qquad ( N \to \infty).
\end{equation}
This was implicitly shown in \cite{MR2430570}, cf. \cite[Remark 5.1]{byun2021real}.

It is obvious from \eqref{X GOE antiGOE} that for $\tau=1$,
\begin{equation} \label{EN VN tau 1}
\mathbb E \mathcal N_N^{(1)}=N, \qquad \Var \mathcal N_N^{(1)}=0. 
\end{equation}
From \eqref{EN tau<1}, \eqref{VN tau strong} and \eqref{EN VN tau 1}, one can expect the occurrence of a non-trivial transition when $\tau \uparrow 1$.
This occurs in the so-called weak non-Hermiticity regime, which was introduced in the pioneering work \cite{MR1431718,fyodorov1997almost,MR1634312} of Fyodorov, Khoruzhenko, and Sommers.  
For the eGinOE, it corresponds to the regime 
\begin{equation} \label{tau AH scaling}
\tau= 1- \frac{\alpha^2}{N}, \qquad \alpha \in (0,\infty)\text{ fixed.}
\end{equation}
This regime is sometimes referred to as the weakly asymmetric regime in the case of real random matrices. Several interesting scaling limits arise in this regime, which interpolate between the GOE and GinOE, see e.g. \cite{MR2430570,AkemannPhillips,FT20} and references therein. 
In this critical regime, it was shown in a recent work \cite{byun2021real} that
\begin{equation}\label{EN tau weak}
\mathbb E \mathcal N_N^{(\tau)}= c(\alpha) \, N +O(1), \qquad ( N \to \infty),  
\end{equation}
where 
\begin{equation} \label{c(alpha)}
\begin{split}
c(\alpha)&:=e^{-\alpha^2/2} \Big[ I_0\Big(\frac{\alpha^2}{2}\Big)+I_1\Big(\frac{\alpha^2}{2}\Big) \Big].  
\end{split}
\end{equation}
Here 
\begin{equation}
I_\nu(x)= \sum_{k=0}^\infty \frac{(x/2)^{2k+\nu}}{ k! \Gamma(\nu+k+1) }
\end{equation}
is the modified Bessel function of the first kind \cite[Chapter 10]{olver2010nist}.
See \eqref{c(alpha) v2} for alternative representations of the function $c(\alpha)$. 
It was also shown in \cite{byun2021real} that 
\begin{equation} \label{VN tau weak}
	\Var \mathcal N_N^{(\tau)} = 2 \Big( c(\alpha)-c(\sqrt{2}\alpha) \Big) N+o(N), \qquad ( N \to \infty).  
\end{equation}
(See \cite{akemann2023product} for analogous results for products of GinOE matrices.)

For the GinOE ($\tau=0$), the central limit theorem for the number of real eigenvalues (or its linear statistics in general) was proved in  \cite{fitzgerald2021fluctuations,MR3612267}. 
In all other cases a central limit theorem was missing, and our first result is on this topic. 

\begin{thm}[\textbf{Central limit theorem for the number of real eigenvalues}] \label{Thm_CLT} $ $ \\
Let $N$ be even. As $N \to \infty$, we have the convergence in distribution 
\begin{equation}
 \frac{\mathcal N_N^{(\tau)} - \mathbb E \mathcal N_N^{(\tau)}  }{\sqrt{ \mathbb E \mathcal N_N^{(\tau)} } } \to  N(0,\sigma^2),
\end{equation}
where $N(0,\sigma^2)$ denotes the normal distribution with mean $0$ and variance
\begin{equation} \label{eq:thmCLT}
\sigma^2= \begin{cases}
2-\sqrt{2} &\textup{if } \tau \in [0,1) \textup{ is fixed},
\smallskip 
\\
\displaystyle 2-2\frac{c(\sqrt{2}\alpha)}{ c(\alpha) } &\textup{if } \tau= 1-\frac{\alpha^2}{N} \textup{ with fixed } \alpha \in (0,\infty).
\end{cases}
\end{equation}
\end{thm}

The variance for the weak non-Hermiticity regime in \eqref{eq:thmCLT} interpolates between the GOE ($\alpha\downarrow 0$) and GinOE ($\alpha\to\infty$). For fixed $\tau$, it can be shown that the results are in fact valid for $-1<\tau<1$, see Corollary \ref{rem:tau>-1}. 
Let us also mention that the full counting statistics of the GinUE and its generalisation were obtained in \cite{fenzl2022precise,MR4458157,byun2022characteristic} with great precision. 
\subsection{Main results on large deviations for the number of real eigenvalues}

We shall now discuss large deviations concerning the number of real eigenvalues of the GinOE and eGinOE. 
We know from \cite{edelman1997probability,MR2430570} that 
\begin{equation}
p_{N,N}^{(\tau)}= \Big( \frac{1+\tau}{2} \Big)^{ \frac{N(N-1)}{4}}. 
\end{equation}
The case of exactly one complex eigenvalue pair was studied in \cite{KanzieperAkemann07} for $\tau=0$, which reads
\begin{align*}
\log p_{N,N-2}^{(0)} = -\frac{\log 2}{4} N^2+\frac{\log(3\sqrt 2)}{2} N + o(N), \qquad (N\to \infty).
\end{align*}
We also mention that the case $k\sim a N$, with $0<a<1$ fixed, was studied in \cite{MR3478312} using a Coulomb gas approach.
In this paper however, we shall be interested in the case of a small number of real eigenvalues. It was shown in \cite{MR3563192} that for the Ginibre case when $\tau=0$, 
\begin{equation} \label{KPTTZ 0}
\lim_{N \to \infty} \frac{1}{ \sqrt{N} } \log p_{N,k_N}^{(0)}=-\frac{1}{ \sqrt{2\pi} } \zeta \Big( \frac32 \Big), 
\end{equation}
whenever $(k_N)_N$ is a sequence of even numbers such that $k_N = o(\frac{\sqrt N}{\log N})$ as $n\to\infty$. 

We prove the analogous statement for the eGinOE. 

\begin{thm}[\textbf{Large deviations for real eigenvalues at strong non-Hermiticity}] \label{thm:mainStrong} $ $ \\
Let $\tau \in [0,1)$ be fixed, and let $N$ and $k$ be even numbers. Then for any fixed $k$  
   \begin{equation}  \label{main strong}
\lim_{N \to \infty} \frac{1}{ \sqrt{N} } \log p_{N,k}^{(\tau)} = - \sqrt{ \frac{1+\tau}{1-\tau} }  \frac{1}{ \sqrt{2\pi} } \zeta \Big( \frac32 \Big). 
\end{equation}
The limit holds with $k$ replaced by any sequence $(k_N)_N$ of even numbers such that $k_N = o(\frac{\sqrt N}{\log N})$ as $N\to\infty$. 
\end{thm}

Theorem \ref{thm:mainStrong} in fact holds for $-1<\tau<1$, see Corollary \ref{rem:tau>-1}.

We also state the analogue of Theorem \ref{thm:mainStrong} for the weak non-Hermiticity regime. Here, due to the lack of a uniform estimate as in Lemma \ref{lem:inequalityTrMNtaum}, we are merely able to give an upper bound. We do believe that this upper bound is in fact sharp, i.e. the inequality in \eqref{main weak} is an equality. In any case, there are already some interesting conclusions that can be drawn from the upper bound, e.g. that the probability of having only a few real eigenvalues is of a much smaller order then in the fixed $\tau$ case. 

\begin{thm}[\textbf{Large deviations for real eigenvalues at weak non-Hermiticity}] \label{thm:mainWeak}
$ $ \\
Let $N$ and $k$ be even and let $\tau=1-\frac{\alpha^2}{N}$ with fixed $\alpha \in (0,\infty)$. Then 
      \begin{equation}  \label{main weak}
\lim_{N \to \infty} \frac{1}{ N } \log p_{N,k}^{(\tau)} \leq -d(\alpha),
\end{equation}
where 
\begin{equation} \label{d(alpha)}
d(\alpha):= \sum_{m=1}^\infty \frac{ c(\sqrt{m}\,\alpha) }{2m} 
=  -\frac{2}{\pi} \int_0^1 \log\Big(1-e^{-\alpha^2 s^2}\Big) \sqrt{1-s^2} \, ds.
\end{equation}
Here $c(\alpha)$ is given by \eqref{c(alpha)}. 
Moreover, the inequality 
holds with $k$ replaced by any sequence $(k_N)_N$ of even numbers such that $k_N = o(\frac{ N}{\log N})$ as $N\to\infty$. 
\end{thm}

The second identity in \eqref{d(alpha)} is shown in Lemma~\ref{Lem_d(alpha) integral rep} below.

\begin{rem}[Interpolating property]
Let us assume that the inequality in \eqref{main weak} is an equality, which is what we expect. 
We can then write Theorems \ref{thm:mainStrong} and \ref{thm:mainWeak} combined as
\begin{align} \label{eq:interpolatinlogpENN}
\lim_{N\to\infty} \frac{\log p_{N,k}^{(\tau)}}{\mathbb E \mathcal N_{N}^{(\tau)}} = -
\begin{cases}
\displaystyle \frac{1}{2} \zeta\Big(\frac{3}{2}\Big) &\textup{if } \tau \in [0,1) \textup{ is fixed},
\smallskip 
\\
\displaystyle \frac{d(\alpha)}{c(\alpha)} &\textup{if } \tau= 1-\frac{\alpha^2}{N} \textup{ with fixed } \alpha \in (0,\infty),
\end{cases}
\end{align}
following directly from \eqref{EN tau<1} and \eqref{EN tau weak}. Indeed, using \eqref{d(alpha)}, we have that
\begin{align*}
\frac{d(\alpha)}{c(\alpha)} = \frac{1}{2} \sum_{m=1}^\infty \frac{1}{m} \frac{c(\sqrt m \alpha)}{c(\alpha)}
\to \begin{cases} \displaystyle
\frac{1}{2} \sum_{m=1}^\infty \frac{1}{m} = \infty, & \alpha\to 0,\\
\displaystyle{\frac{1}{2} \sum_{m=1}^\infty \frac{1}{m\sqrt m} = \frac{1}{2}} \zeta\Big(\frac{3}{2}\Big), & \alpha\to \infty,
\end{cases}
\end{align*}
which can straightforwardly be derived from the representation \eqref{c(alpha) v2} for $c(\alpha)$. We thus observe an interpolation between the GinOE and the GOE ($p^{(1)}_{N,k}=0$ if $k < N$). 
\end{rem}

In the form \eqref{eq:interpolatinlogpENN}, the fixed $\tau$ limit does not depend on $\tau$ anymore. This might indicate a universality result. We can consider real matrices $M$, distributed by
\begin{align*}
    \displaystyle\frac{1}{Z_N^{(V)}} e^{-\frac{1}{2} \Tr V(M)} dM_N,
\end{align*}
for some external field $V$. It is an interesting question whether the probabilities $p_{N,k}^{(V)}$ of having $k$ real eigenvalues satisfy
\begin{align*}
\lim_{n\to\infty} \frac{\log p_{N,k}^{(V)}}{\mathbb E \mathcal N_N^{(V)}} = - \frac{1}{2}\zeta\Big(\frac{3}{2}\Big)
\end{align*}
for a general class of external fields $V$ (when $N$ and $k$ have the same parity, and $k$ growing sufficiently slowly with $N$). 

\subsection{Further results}  \label{Subsection_further results}

We now present some further results. Let $\operatorname{Li}_{s}(z)$ denote the polylog function of order $s$ defined according to the Dirichlet series
\begin{equation}
\operatorname{Li}_{s}(z) = \sum_{k=1}^{\infty}\frac{z^{k}}{k^{s}},
\end{equation}
which is an analytic function of $z$ on the open unit disc $|z|<1$. It can be analytically continued to $\mathbb C\setminus [1,\infty]$, and this analytic continuation can be continuously extended to $\mathbb C\setminus(1,\infty)$ when $\re s>1$. When $z=1$ and $\re s>1$ it coincides with the Riemann zeta function: $\operatorname{Li}_{s}(1) = \zeta(s)$. 
\begin{thm} \label{cor:logdetPolylog}
Let $N=2n$ be even, and let $\tau \in [0,1)$ be fixed. 
Then we have for all $x\in[0,2]$ that
\begin{align} \label{eq:cor:logdetPolylog1}
\lim_{N\to\infty} \frac{1}{\sqrt{N}}\log \bigg( \sum_{k=0}^{n} p_{N,2k}^{(\tau)}x^{k} \bigg) = 
-\sqrt{\frac{1+\tau}{1-\tau}}\frac{\operatorname{Li}_{3/2}(1-x)}{\sqrt{2\pi}}.
\end{align}
The convergence is uniform on any compact subset of $(0,2)$. 
\end{thm}

The uniform convergence cannot be extended to $x=0$, since that would imply that the limit function is smooth at $x=0$.

\begin{rem} \label{rem:generatingSumpNk}
Expanding around $x=1$, Theorem \ref{cor:logdetPolylog} gives us that  
\begin{align*}
    \lim_{N\to\infty} \frac{1}{\sqrt{N}} \log\bigg(1+\sum_{k=1}^n \sum_{j=k}^n \binom{j}{k} p_{N,2j}^{(\tau)}  (x-1)^k\bigg)
    = -\sqrt{\frac{1+\tau}{1-\tau}} \frac{\operatorname{Li}_{3/2}(1-x)}{\sqrt{2\pi}}
\end{align*}
uniformly for $x$ in compact subsets of $(0,2)$. 
This implies that all the coefficients on the left-hand side converge to some limit as $N\to\infty$. 
We have 
\begin{align*}
& \quad  \frac{1}{\sqrt{N}} \log\bigg(1+\sum_{k=1}^n \sum_{j=k}^n \binom{j}{k} p_{N,2j}^{(\tau)}  (x-1)^k\bigg)
\\
&= \frac{x-1}{\sqrt{N}} \sum_{j=1}^n j p_{N,2j}^{(\tau)}   -\frac{(x-1)^2}{\sqrt{N}}\bigg(\sum_{j=2}^n \frac{j(j-1)}{2} p_{N,2j}^{(\tau)} - \frac{1}{2}\sum_{j,\ell=1}^n j\ell p_{N,2j}^{(\tau)}  p_{N,2\ell}^{(\tau)} \bigg)
    +\mathcal O((x-1)^3).
\end{align*}
We can use this as a generating function for various probabilistic expressions, by taking derivatives at $x=1$.
For example, the expected number of real eigenvalues is asymptotically given by 
\begin{align*}
    \lim_{N\to\infty} \frac{1}{\sqrt{N}} \mathbb E \mathcal N_N^{(\tau)}
    = \lim_{N\to\infty} \frac{1}{\sqrt{N}} \sum_{k=0}^n 2j p_{N,2j}^{(\tau)} 
    = 2\sqrt{\frac{1+\tau}{1-\tau}} \frac{\operatorname{Li}_{3/2}'(0)}{\sqrt{2\pi}}
    = 2\sqrt{\frac{1+\tau}{1-\tau}} \frac{1}{\sqrt{2\pi}}.
\end{align*}
\end{rem}

\begin{cor} \label{rem:tau>-1}
The implications of Theorem \ref{Thm_CLT} and Theorem \ref{thm:mainStrong} are true for fixed $-1<\tau<1$.
\end{cor}

\begin{proof}
The skew-orthogonal polynomials in \eqref{SOP} are also valid for $-1<\tau<0$. This implies in particular that the generating identity in Proposition \ref{Prop_finite N} is also valid for $-1<\tau<0$. We may write
\begin{align*}
 \frac{1}{\sqrt{N}}\log \bigg( \sum_{k=0}^{n} p_{N,2k}^{(\tau)}x^{k} \bigg) = \sum_{j=0}^N C_{N,j}(x) \tau^j,   
\end{align*}
for some coefficients $C_{N,j}(x)$, for all $-1<\tau<1$. Theorem \ref{cor:logdetPolylog} means for the coefficients that we have
\begin{align*}
    \lim_{N\to\infty} C_{N,j}(x) = 
    -\frac{\operatorname{Li}_{3/2}(1-x)}{\sqrt{2\pi}}
    \frac{1}{j!}\frac{d^j}{d\tau^j} \Bigg|_{\tau=0} \sqrt{\frac{1+\tau}{1-\tau}},
\end{align*}
and \eqref{eq:cor:logdetPolylog1} is in particular then also true for $-1<\tau<0$ (at least for fixed $x$). Thus Theorem \ref{thm:mainStrong} follows directly (take $x=0$), and along the lines of Remark \ref{rem:generatingSumpNk}, we obtain Theorem \ref{Thm_CLT} as well.
\end{proof}

\begin{cor}
    Let $N$ and $(k_N)_N$ be even numbers. We have as $N\to\infty$ that
    \begin{align*}
        \frac{\log p_{N,k_N}^{(\tau)}}{\sqrt N} \leq  - \Big(1-\frac{1}{\sqrt 2}\Big)\sqrt{\frac{1+\tau}{1-\tau}} \frac{\zeta(3/2)}{\sqrt{2\pi}}
        - \frac{k_N}{\sqrt N}\log 2+o(1).
    \end{align*}
    If $k=k_N^{(\tau)}$ is such that $p_{N,k}^{(\tau)}$ is maximal among $p_{N,0}^{(\tau)}, p_{N,2}^{(\tau)}, \ldots, p_{N,N}^{(\tau)}$, then we have
    \begin{align*}
        \lim_{N\to\infty} \frac{k_N^{(\tau)}}{\sqrt N} \leq \sqrt{\frac{1+\tau}{1-\tau}}\frac{\zeta(3/2)}{\sqrt\pi \log 4}.
    \end{align*}
\end{cor}

\begin{proof}
    Taking $x=2$ in Theorem \ref{cor:logdetPolylog}, we find that
    \begin{align*}
        \lim_{N\to\infty} \frac{\log\left( p_{N,k_N}^{(\tau)} 2^{k_N}\right)}{\sqrt N}
        \leq 
        \lim_{N\to\infty} \frac{1}{\sqrt N}\log \sum_{k=0}^n p_{N,2k}^{(\tau)} 2^{k} = -\sqrt{\frac{1+\tau}{1-\tau}}\frac{\operatorname{Li}_{3/2}(-1)}{\sqrt{2\pi}}
        = \Big(1-\frac{1}{\sqrt 2}\Big) \sqrt{\frac{1+\tau}{1-\tau}} \frac{\zeta(3/2)}{\sqrt{2\pi}},
    \end{align*}
    from which the first assertion follows. This, combined with Theorem \ref{thm:mainStrong}, yields the estimate for $k_N^{(\tau)}$ as $N\to\infty$.
\end{proof}


The rest of this paper is organised as follows. 
\begin{itemize}
    \item In Section~\ref{Section_proof main}, we introduce key ingredients of our analysis and complete the proof of the main results.
    In Subsection~\ref{Subection_proof outline}, we present 
    Propositions~\ref{Prop_finite N}, \ref{Prop_generating matrix implementation}, \ref{prop:asympTracem}, \ref{prop:asympTracem weak} and Lemmas~\ref{lem:inequalityTrMNtaum}, \ref{Lem_eigenvalue of Mn tau boundedness}, \ref{cor:remainder} some of which will be shown in the following sections. Combining all of these, Subsection~\ref{Subsection_proof main} culminates in the proof of Theorems~\ref{Thm_CLT}, \ref{thm:mainStrong} and \ref{thm:mainWeak}. 
    \smallskip 
    \item Section~\ref{Section_generating matrix} is devoted to the analysis of the generating function of the number of real eigenvalues.
    In Subsections~\ref{Subsection_generating function} and \ref{Subsection_generating matrix evaluation}, we prove the finite-$N$ result (Proposition~\ref{Prop_finite N}) and provide some useful lemmas on the evaluations of the generating matrix. 
    In Subsections~\ref{Subsection_generating matrix estimates} and \ref{Subsection_error estimates lemmas}, we show some preliminary estimates of the generating matrix and prove Lemmas~\ref{lem:inequalityTrMNtaum} and \ref{Lem_eigenvalue of Mn tau boundedness}.
    \smallskip 
    \item Sections~\ref{Section_asymptotic strong} and \ref{Section_asymptotic weak} are devoted to the crucial asymptotic analysis of the generating matrix at strong and weak non-Hermiticity, respectively.
    In particular, we show Propositions~\ref{prop:asympTracem} and \ref{prop:asympTracem weak}, which complete the proofs of main results.    \smallskip 
    \item This article contains two appendices. 
    In Appendix~\ref{Appendix A_auxiliary lemmas}, we collect some auxiliary lemmas. In Appendix~\ref{Appendix B_equivalent det formula}, we introduce an equivalent determinantal formula of the generating function due to Forrester and Nagao, and compare it with Proposition~\ref{Prop_finite N}.  
\end{itemize}

\subsection*{Acknowledgements}

SB is partially supported by Samsung Science and Technology Foundation (SSTF-BA1401-51), by a KIAS Individual Grant (SP083201) via the Center for Mathematical Challenges at Korea Institute for Advanced Study, by the National Research Foundation of Korea (NRF-2019R1A5A1028324), and by the POSCO TJ Park Foundation (POSCO Science Fellowship).
LDM is funded by the Deutsche Forschungsgemeinschaft (DFG, German Research Foundation) – SFB 1283/2 2021– 317210226 ``Taming uncertainty and profiting from randomness and low regularity in analysis, stochastics and their applications", and the Royal Society grant RF\textbackslash ERE\textbackslash 210237. NS acknowledges financial support from the Royal Society grant URF\textbackslash R1\textbackslash180707.

\section{Proofs of main results} \label{Section_proof main}

In this section and later sections, we shall assume that $N=2n$, where $n$ is a positive integer. In the proceeding, we shall mostly express our results and proofs in terms of $n$ rather than $N$.

For reader's convenience, we briefly explain the overall strategy of the proof.
\begin{enumerate}
    \item[(i)] We first derive a determinantal formula for the generating function of the number of real eigenvalues (Proposition~\ref{Prop_finite N}) that holds for any $\tau$ and $n$. 
    The generating matrix appearing in Proposition~\ref{Prop_finite N} can be implemented to express the probability that there is no real eigenvalue and cumulants of the number of real eigenvalues (Proposition~\ref{Prop_generating matrix implementation}). 
    \smallskip
    \item[(ii)] We then derive asymptotic behaviours of trace powers of the generating matrix both in the strong (Proposition~\ref{prop:asympTracem}) and weak (Proposition~\ref{prop:asympTracem weak}) non-Hermiticity. 
    Together with Proposition~\ref{Prop_generating matrix implementation} (ii), these lead to Theorem~\ref{Thm_CLT}. 
    \smallskip 
    \item[(iii)] To complete the proof of Theorems~\ref{thm:mainStrong} and \ref{thm:mainWeak}, we perform required error estimates (Lemmas~\ref{lem:inequalityTrMNtaum} and \ref{Lem_eigenvalue of Mn tau boundedness}). 
\end{enumerate}

\subsection{Key ingredients} \label{Subection_proof outline}
The first step of the proofs is a determinant formula for $p_{2n,2k}$. 
For this purpose, recall that the $k$-th Hermite polynomial $H_k$ is given by
\begin{equation}
H_k(z):=(-1)^k e^{z^2} \frac{d^k}{dz^k} e^{-z^2}=k! \sum_{m=0}^{ \lfloor k/2 \rfloor } \frac{(-1)^m}{ m! (k-2m)! } (2z)^{k-2m}. 
\end{equation}
and that the (regularised) hypergeometric function is defined by the Gauss series
\begin{align}
{}_2 F_1(a,b;c;z):=\frac{\Gamma(c)}{\Gamma(a)\Gamma(b)} \sum_{s=0}^\infty \frac{\Gamma(a+s) \Gamma(b+s) }{ \Gamma(c+s) s! } z^s, \quad (|z|<1)
\end{align}
and by analytic continuation elsewhere.

\begin{prop}[\textbf{Determinantal formula for the generating function}]\label{Prop_finite N}
We have 
\begin{equation} \label{generating function det}
\sum_{k=0}^{n} z^k p_{2n,2k}^{(\tau)}= \det\Big[ \delta_{j,k}+(z-1) M_n^{(\tau)}(j,k)  \Big]_{j,k=1}^n,
\end{equation}
where 
\begin{align}
\label{Mn tau jk}
\begin{split}
M_n^{(\tau)}(j,k) &= \frac{1}{ \sqrt{2\pi} } \frac{  ( \tau/2 )^{j+k-2} }{ \sqrt{ \Gamma(2j-1)\Gamma(2k-1) } }  \int_\R e^{ -\frac{x^2}{1+\tau} } H_{2j-2}\Big( \frac{x}{ \sqrt{2\tau} } \Big) H_{2k-2}\Big( \frac{x}{ \sqrt{2\tau} } \Big)   \,dx
\\
&= \frac{1}{ \sqrt{2\pi} }  \Big( \frac{1+\tau}{1-\tau} \Big)^{\frac12} \frac{ \Gamma(j+k-\frac32){}_2F_1(k-j+\frac{1}{2},j-k+\frac{1}{2};-j-k+\frac52;-\frac{\tau}{1-\tau}) }{ \sqrt{ \Gamma(2j-1) \Gamma(2k-1) } } . 
\end{split}
\end{align}
In particular, we have 
\begin{equation}
\begin{split} \label{p 2n 0 Mn tau}
 p_{2n,0}^{(\tau)}   = \det \Big[ I-M_n^{(\tau)} \Big], \qquad M_n^{(\tau)}=\Big[ M_n^{(\tau)}(j,k) \Big]_{j,k=1}^n. 
\end{split}
\end{equation}
\end{prop} 
For $\tau=0$, a similar formula was first derived by Kanzieper and Akemann \cite{MR2185860}. We stress that the determinantal formula \eqref{generating function det} is equivalent to Proposition~\ref{Prop_FN finite N} due to Forrester and Nagao. 
In general, the determinantal formula for the generating function of number of real eigenvalues follows from the skew-orthogonal polynomial formalism of the generalised partition functions \cite{forrester2007eigenvalue,forrester2013skew}. 
Let us also mention that similar formulas can be found in the context of induced Ginibre, spherical Ginibre, and truncated orthogonal matrices, see e.g.  \cite{forrester2012pfaffian,fischmann2012induced,fischmann2011one,forrester2010limiting,khoruzhenko2010truncations,MR4364735,MR4134827,MR3685239} and also \cite[Section 4]{byun2023progress} for a comprehensive review.

\begin{rem}[Extremal cases $\tau=0,1$]
For $\tau=0$, using that ${}_2 F_1(a,b;c;0)=1$ and that $H_{k}(x)$ has leading coefficient $2^{k}$ we get 
\begin{equation} \label{Hermite to monomial}
\Big( \frac{\tau}{2} \Big)^{k/2} H_k \Big( \frac{x}{ \sqrt{2\tau} } \Big) \to x^k, \qquad \tau \to 0,
\end{equation}
and consequently
\begin{align}
\label{Mn 0 jk}
M_n^{(0)}(j,k) &= \frac{1}{ \sqrt{2\pi} }  \frac{ \Gamma(j+k-\frac32) }{ \sqrt{ \Gamma(2j-1) \Gamma(2k-1) } } 
= \frac{1}{ \sqrt{2\pi} }  \int_0^\infty  \frac{ e^{ -x } }{ x^{5/2} } \frac{ x^{j} }{ \sqrt{\Gamma(2j-1)} } \frac{ x^k }{ \sqrt{\Gamma(2k-1)} } \,dx . 
\end{align}
Thus one can observe that for $\tau=0$, Proposition~\ref{Prop_finite N} corresponds to \cite[Lemma 2.1]{MR3563192}, see also \cite[Proposition 1]{forrester2015diffusion}.
The second expression also agrees with the integral representation in \cite[Eq.(A.19)]{MR3563192}. 
We mention that for $\tau=0$, this follows from exact formulas in \cite{KanzieperAkemann07,MR2185860,forrester2007eigenvalue}, cf. see \cite{borodin2007note} for more on the Pfaffian integration theorem. 

On the other hand, if $\tau=1$, it follows from the orthogonality of the Hermite polynomials
\begin{equation}
\int_\R H_n(x) H_m(x) e^{-x^2}\,dx= \sqrt{\pi} \, 2^n \, n! \, \delta_{nm}
\end{equation}
that $M_n^{(1)}(j,k) = \delta_{jk}.$
Thus one can observe that 
\begin{equation}
\sum_{k=0}^{n} z^k p_{2n,2k}(1)= \det[ zI  ]= z^n, \qquad \textup{i.e.} \quad 
p_{2n,2k}=\begin{cases}
    0 &\textup{if } k=0,1,\dots, n-1,
    \\
    1 & \textup{if } k=n 
\end{cases}
\end{equation}
as expected. 
\end{rem}


\begin{prop}[\textbf{Zero probabilities and cumulants in terms of the generating matrix}] \label{Prop_generating matrix implementation}
For any $\tau$ and $n$, we have the following. 
\begin{itemize}
    \item[(i)] For any natural number $K_n \ge 1$, we have
\begin{equation} \label{log p2n 0 Trunc}
\log p_{2n,0}^{(\tau)}=\Tr \log (I-M_n^{(\tau)})=- \sum_{m=1}^{K_n} \frac{1}{m} \Tr (M_n^{(\tau)})^m - R_n(K_n),
\end{equation}
where
\begin{equation} \label{Rn K integral}
R_n(K)=\int_0^1 \Tr\bigg(  \frac{ (M_n^{(\tau)})^{K+1}  }{  (1-x M_n^{(\tau)})^{K+1}  }  \bigg) (1-x)^K \,dx.  
\end{equation}
\item[(ii)] The $l$-th order cumulant $\kappa_l$ of $\mathcal{N}_{2n}^{(\tau)}$ is given by 
\begin{equation}
\kappa_l= 2^l \sum_{m=1}^l \frac{(-1)^{m+1} }{ m }   \sum_{\substack{\nu_1+\dots+ \nu_m=l \\ \nu_j \ge1 }} \frac{l!}{\nu_1 ! \dots \nu_m !} \Tr (M_n^{(\tau)})^m.
\end{equation}
\end{itemize}
\end{prop}
\begin{proof}
  The first assertion is an immediate consequence of Proposition~\ref{Prop_finite N}, whereas the second one was shown in \cite[Lemma 2.4]{MR3612267}.  
\end{proof}

We need to show the following. 

\begin{prop}[\textbf{Asymptotics of trace powers at strong non-Hermiticity}] \label{prop:asympTracem}
For a fixed $\tau \in [0,1)$, and for any fixed integer $m>0$, 
\begin{equation} \label{Tr M power conv strong}
\lim_{n \to \infty} \frac{1}{ \sqrt{2n} } \Tr (M_n^{(\tau)})^m = \sqrt{ \frac{1+\tau}{1-\tau}  \frac{1}{2\pi m} }.
\end{equation}
\end{prop}

For $\tau=0$, this proposition coincides with \cite[Lemma 2.3]{MR3563192}.

\begin{prop}[\textbf{Asymptotics of trace powers at weak non-Hermiticity}]  \label{prop:asympTracem weak}
For $\tau= 1-\alpha^2/(2n)$ with a fixed $\alpha \in (0,\infty)$, and for any fixed integer $m>0$, 
\begin{equation} \label{Tr M power conv weak}
\begin{split}
\lim_{n \to \infty} \frac{1}{ 2n } \Tr (M_n^{(\tau)})^m & = \frac{c(\sqrt{m}\,\alpha)}{2}= \frac{e^{-m\alpha^2/2}}{2}  \Big[ I_0\Big(\frac{m\alpha^2}{2}\Big)+I_1\Big(\frac{m\alpha^2}{2}\Big) \Big]. 
\end{split}
\end{equation}
\end{prop}

\begin{ex}[The first three cumulants]
As a consequence of Proposition~\ref{Prop_generating matrix implementation} (ii), for $l=1,2,3$, we have 
\begin{align}
 \label{Tr expectation}
 \mathbb E \mathcal N_{2n}^{(\tau)} &= 2 \Tr M_n^{(\tau)}, 
 \\
 \label{Tr variance}
	\Var \mathcal{N}_{2n}^{(\tau)} & = 4 \Big( \Tr  M_n^{(\tau)}  -\Tr  (M_n^{(\tau)})^2 \Big) ,
 \\
\mathbb{E} [ (\mathcal{N}_{2n}^{(\tau)}-\mathbb E \mathcal N_{2n}^{(\tau)})^3 ] & =8\Big( \Tr M_n^{(\tau)}-3 \Tr (M_n^{(\tau)})^2  + 2 \Tr(M_n^{(\tau)})^3 \Big) . 
\end{align}
We emphasise that for $m=1,2$, Propositions~\ref{prop:asympTracem} and ~\ref{prop:asympTracem weak} recover important results in \cite{MR2430570,byun2021real}. 
To be more precise, for $m=1$, the asymptotic behaviours of the expected numbers \eqref{EN tau<1} and \eqref{EN tau weak} follow from \eqref{Tr expectation}, whereas for $m=2$, the variance asymptotics \eqref{VN tau strong} and \eqref{VN tau weak} follow from \eqref{Tr variance}.  
\end{ex}

To prove Theorem \ref{thm:mainStrong}, we need a version of Proposition \ref{prop:asympTracem} that is uniform in $m$, rather than for fixed $m$. We shall need the following more refined bound.

\begin{lem} \label{lem:inequalityTrMNtaum}
    For any $\tau\in[0,1)$ and integers $n\geq 1$ and $m \geq \frac{1+\tau}{1-\tau}$ we have the inequality
    \begin{equation}
      \Tr (M_n^{(\tau)})^m 
        \leq \frac{1}{4} + \sqrt{\frac{1+\tau}{1-\tau}} \sqrt{\frac{n}{\pi m}} (1+2/n)
        + \frac{1}{8} \frac{1-\tau}{1+\tau} \sqrt{\frac{m}{\pi n}} (1+1/n).
    \end{equation}
\end{lem}

Note that for $\tau=0$, this lemma gives  \cite[Lemma 2.3]{MR3563192}.

\begin{lem}[\textbf{Bounds for eigenvalues of $M_n^{(\tau)}$}] \label{Lem_eigenvalue of Mn tau boundedness}
Let $\lambda_1>\lambda_2 > \dots >\lambda_n$ $(j=1,2,\dots,n)$ be the eigenvalues of $M_n^{(\tau)}$.  
Then we have the following. 
\begin{enumerate}
    \item[(i)] The matrix $M_n^{(\tau)}$ is positive definite, i.e. $\lambda_n>0$.
    \item[(ii)] There exists $\mu>0$ such that for sufficiently large $n$, 
\begin{equation} \label{lambda1 upper bdd}
\lambda_1 \le 1-\frac{ \mu }{ n^a }, \qquad  a= \begin{cases}
    1 &\textup{if }\tau \textup{ is fixed},
    \\
    2 &\textup{if }\tau= 1-\frac{\alpha^2}{2n}. 
\end{cases}
\end{equation}
\end{enumerate}
\end{lem}

\begin{lem}\label{cor:remainder}
Let $a$ be given in \eqref{lambda1 upper bdd}. 
Then we have 
\begin{equation}
\frac{1}{(2n)^{a/2}}|R_n(K)|  \leq \frac{1}{\sqrt{2^a\mu}} \frac{1}{\sqrt{K+1}} \Tr\left((M_n^{(\tau)})^{K+1}\right).
\end{equation}
\end{lem}
\begin{proof}
By \eqref{Rn K integral} and Lemma~\ref{Lem_eigenvalue of Mn tau boundedness}, it suffices to show that for $0<\lambda<1$,
\begin{equation}
 \lambda^{K+1} \int_0^1 \frac{(1-x)^K}{(1-\lambda x)^{K+1}} dx
        = -\log(1-\lambda) - \sum_{j=1}^K \frac{\lambda^j}{j}
        \leq \frac{\lambda^{K+1}}{\sqrt{2K+1}} \frac{1}{\sqrt{1-\lambda}}.
\end{equation}
For this, note that 
\begin{align*}
    \lambda^{K+1} \int_0^1 \frac{(1-x)^K}{(1-\lambda x)^{K+1}} dx
    &= \int_{0}^\lambda \frac{(\lambda-x)^K}{(1-x)^{K+1}} dx = \int_{1-\lambda}^1 \frac{(\lambda-1+x)^K}{x^{K+1}} dx\\
    &= \int_{1-\lambda}^1 \sum_{j=0}^K \binom{K}{j} (\lambda-1)^j x^{-j-1} dx\\
    &= -\log(1-\lambda) 
    - \sum_{j=1}^K \frac{1}{j} \binom{K}{j} \left((\lambda-1)^j-(-1)^j\right)= -\log(1-\lambda) - \sum_{j=1}^K \frac{\lambda^j}{j}.
\end{align*}
Here, the identity 
$$
   \sum_{j=1}^K \frac{1}{j} \binom{K}{j} \left((\lambda-1)^j-(-1)^j\right)= \sum_{j=1}^K \frac{\lambda^j}{j} 
$$
follows from
$$
  \sum_{j=1}^K  \binom{K}{j} (\lambda-1)^{j-1}=
  \frac{(1+\lambda-1)^K-1}{\lambda-1}
  =\sum_{j=1}^K \lambda^{j-1}. 
$$
The inequality follows by Cauchy-Schwarz, as follows.
\begin{align*}
    -\log(1-\lambda)-\sum_{j=1}^K \frac{\lambda^j}{j}
   & = \sum_{j=K+1}^\infty \frac{\lambda^j}{j}
    = \int_0^\lambda \frac{x^K}{1-x} dx
   \\
   &\leq  \sqrt{\int_0^\lambda x^{2K} dx \int_0^\lambda \frac{dx}{(1-x)^2}}
    = \frac{\lambda^{K+1/2}}{\sqrt{2K+1}} \frac{\sqrt \lambda}{\sqrt{1-\lambda}}.
\end{align*}
\end{proof}

\subsection{Proofs of Theorems~\ref{Thm_CLT}, \ref{thm:mainStrong}, \ref{thm:mainWeak} and \ref{cor:logdetPolylog}} \label{Subsection_proof main}

This section culminates in the proofs of our main results.

\begin{proof}[Proof of Theorem~\ref{Thm_CLT}]
As a consequence of Propositions~\ref{prop:asympTracem} and ~\ref{prop:asympTracem weak}, for any fixed $m$, we have that $\Tr (M_n^{(\tau)})^m$ is of order $n^{a/2}$, where $a$ is given by \eqref{lambda1 upper bdd}.
Then it follows from Proposition~\ref{Prop_generating matrix implementation} (ii) that all the cumulants have order $n^{a/2}$. 
This shows the desired central limit theorem. 
\end{proof}

\begin{proof}[Proof of Theorem~\ref{thm:mainStrong}]
 It is enough to show Theorem~\ref{thm:mainStrong} for $k=0$: 
    \begin{equation*}
        \lim_{n \to \infty} \frac{1}{ \sqrt{2n} } \log p_{2n,0}^{(\tau)}  = - \sqrt{ \frac{1+\tau}{1-\tau} }  \frac{1}{ \sqrt{2\pi} } \zeta \Big( \frac32 \Big).
    \end{equation*}
    The statement for general $k$ can be proved entirely analogously to \cite[Section 2.2]{MR3563192}, which only needs the result of Lemma \ref{Lem_eigenvalue of Mn tau boundedness} and Proposition \ref{prop:asympTracem}.

Let $m_\tau\geq \frac{1+\tau}{1-\tau}$ be an integer. For any positive integer $K>m_\tau$, we can use the inequality from Lemma \ref{lem:inequalityTrMNtaum} to show that
\begin{align*}
    \frac{1}{\sqrt{2n}}\sum_{m=m_\tau}^K \frac{1}{m} \Tr (M_n^{(\tau)})^m 
    \leq \sqrt{\frac{1+\tau}{1-\tau}}\frac{1+2/n}{\sqrt{2\pi}}  \sum_{m=m_\tau}^K \frac{1}{m^{3/2}}
    + \frac{\log K}{4\sqrt{2n}}
    + \frac{\sqrt K}{8 n \sqrt{2\pi}} \frac{1-\tau}{1+\tau}.
\end{align*}
Now we take $K=K_n$ such that both $n^{-1/2} K_n\to\infty$ and $K_n = \mathcal O(n \log^2 n)$ as $n\to\infty$. Then we obtain
\begin{align*}
    \frac{1}{\sqrt{2n}}\sum_{m=m_\tau}^{K_n} \frac{1}{m} \Tr (M_n^{(\tau)})^m \leq 
    \sqrt{\frac{1+\tau}{1-\tau}}\frac{1+2/n}{\sqrt{2\pi}}  \sum_{m=m_\tau}^K \frac{1}{m^{3/2}} + C \frac{\log n}{\sqrt n}
\end{align*}
for some constant $C>0$. By Lemma~\ref{cor:remainder} the remainder satisfies
\begin{align*}
    \frac{1}{\sqrt{2n}} |R_n(K_n)| \leq c \frac{\sqrt n}{K_n}
\end{align*}
for some constant $c>0$. We infer that
$$
  -\frac{1}{\sqrt{2n}}\log p_{2n,0}^{(\tau)} 
    \leq \sum_{m=1}^{m_\tau-1} \frac{1}{m} \Tr (M_n^{(\tau)})^m +
    \sqrt{\frac{1+\tau}{1-\tau}}\frac{1+2/n}{\sqrt{2\pi}}  \sum_{m=m_\tau}^{K_n} \frac{1}{m^{3/2}} + C \frac{\log n}{\sqrt n}
    + c \frac{\sqrt n}{K_n}.
$$
For any fixed number $K$, we have
\begin{align*}
   -\frac{1}{\sqrt{2n}}\log p_{2n,0}^{(\tau)}  \geq  
   \frac{1}{\sqrt{2n}}\sum_{m=1}^{K} \frac{1}{m} \Tr (M_n^{(\tau)})^m.
\end{align*}
Taking the limit $n\to\infty$, and using Proposition \ref{prop:asympTracem}, we find that
\begin{align*}
    \sqrt{\frac{1+\tau}{1-\tau}}\frac{1}{\sqrt{2\pi}}  \sum_{m=1}^K \frac{1}{m^{3/2}} \leq 
    -\lim_{n\to\infty} \frac{1}{\sqrt{2n}}\log p_{2n,0}^{(\tau)}  \leq
    \sqrt{\frac{1+\tau}{1-\tau}} \frac{\zeta(3/2)}{\sqrt{2\pi}}.
\end{align*}
Since this is true for any positive integer $K$, we conclude that
\begin{align*}
\lim_{n\to\infty} \frac{1}{\sqrt{2n}}\log p_{2n,0}^{(\tau)}  = -\sqrt{\frac{1+\tau}{1-\tau}}\frac{\zeta(3/2)}{\sqrt{2\pi}}.
\end{align*}
\end{proof}

\begin{proof}[Proof of Theorem~\ref{thm:mainWeak}]
   For $k=0$, this immediately follows from Lemma~\ref{Lem_eigenvalue of Mn tau boundedness} and Proposition~\ref{prop:asympTracem weak}. 
   Furthermore, the statement for general $k$ can be again proved along the same lines to \cite[Section 2.2]{MR3563192}, which only needs the result of Lemma \ref{Lem_eigenvalue of Mn tau boundedness} and Proposition \ref{prop:asympTracem weak}. 
\end{proof}

\begin{proof}[Proof of Theorem \ref{cor:logdetPolylog}]
    We already proved the statement for $x=0$ in Theorem~\ref{thm:mainStrong}. By Lemma \ref{lem:inequalityTrMNtaum}, for every $x\in(0,1]$, we have for $m_\tau\geq\frac{1+\tau}{1-\tau}$ that
    \begin{align}
    \begin{split}
    &\quad \frac{1}{\sqrt{2n}}\sum_{m=m_\tau}^\infty  \Tr (M_n^{(\tau)})^m \frac{(1-x)^m}{m}
    \\ 
    &\leq \sqrt{\frac{1+\tau}{1-\tau}}\frac{1+2/n}{\sqrt{2\pi}}  \sum_{m=m_\tau}^\infty \frac{(1-x)^m}{m^{3/2}} + \frac{1}{4\sqrt{2n}} \sum_{m=1}^\infty \frac{(1-x)^m}{m}
    + \frac{1-\tau}{1+\tau} \frac{1+1/n}{8 n \sqrt{2\pi}} \sum_{m=1}^\infty \frac{(1-x)^m}{\sqrt m}
    \\ \label{eq:estimateWith1-x}
    &= \sqrt{\frac{1+\tau}{1-\tau}}\frac{1+2/n}{\sqrt{2\pi}}  \sum_{m=m_\tau}^\infty \frac{(1-x)^m}{m^{3/2}}
    + \frac{\log 1/x}{4\sqrt{2n}}
    +\sqrt{\frac{1+\tau}{1-\tau}}\frac{\operatorname{Li}_{1/2}(1-x)}{8n\sqrt{2\pi}}(1+1/n).   
    \end{split}
\end{align}
Hence for every integer $K>0$ we have
\begin{align*}
    \sqrt{\frac{1+\tau}{1-\tau}} \frac{1}{\sqrt{2\pi}} \sum_{m=1}^K \frac{(1-x)^m}{m^{3/2}}
    \leq -\lim_{n\to\infty} \log \bigg(\sum_{k=0}^{2n} p_{2n,2k}^{(\tau)}  x^{2k}\bigg)
    \leq \sqrt{\frac{1+\tau}{1-\tau}}\frac{\operatorname{Li}_{3/2}(1-x)}{\sqrt{2\pi}}.
\end{align*}
Since this is true for any $K$, we conclude that for any $x\in(0,1]$
\begin{align*}
    \lim_{n\to\infty} \log \bigg(\sum_{k=0}^{2n} p_{2n,2k}^{(\tau)} x^{2k}\bigg)
    = -\sqrt{\frac{1+\tau}{1-\tau}}\frac{\operatorname{Li}_{3/2}(1-x)}{\sqrt{2\pi}}.
\end{align*}
Furthermore, we notice that for all $x\in[0,1]$
\begin{align*}
    \sum_{m=K+1}^\infty \frac{(1-x)^m}{m^{3/2}}
    \leq \int_K^\infty \frac{(1-x)^m}{m^{3/2}} dm
    = \frac{(1-x)^K}{\sqrt K} + \sqrt \pi \sqrt{-\log(1-x)} \erfc(\sqrt{- K \log(1-x)}),
\end{align*}
where we read the expression on the right-hand side as a limit (which is $0$) when $x=1$. This, combined with \eqref{eq:estimateWith1-x}, gives a uniform bound for all $x\in [r,1]$ for any fixed $r>0$.

Let us now turn to the case $x\in[1,2]$. We can write
\begin{align*}
    \sum_{m=1}^\infty \Tr (M_n^{(\tau)})^m \frac{(1-x)^m}{m}
    &= \sum_{m=1}^\infty (-1)^m \Tr (M_n^{(\tau)})^m \frac{|1-x|^m}{m}\\
    &= \sum_{m=1}^\infty \Tr (M_n^{(\tau)})^m \frac{|1-x|^m}{m}
    - 2 \sum_{m=1}^\infty \Tr (M_n^{(\tau)})^{2m} \frac{|1-x|^{2m}}{2m}.
\end{align*}
Both series in the last line can be treated with the same arguments that we used for the $x\in(0,1]$ case, yielding
\begin{align*}
    \lim_{n\to\infty} \frac{1}{\sqrt{2n}} \sum_{m=1}^\infty \Tr (M_n^{(\tau)})^m \frac{(1-x)^m}{m}
     &= \sqrt{\frac{1+\tau}{1-\tau}} \frac{1}{\sqrt{2\pi}}
    \left(\operatorname{Li}_{3/2}(|1-x|)-\frac{1}{\sqrt 2} \operatorname{Li}_{3/2}(|1-x|^2)\right)\\
    &= \sqrt{\frac{1+\tau}{1-\tau}} \frac{\operatorname{Li}_{3/2}(1-x)}{\sqrt{2\pi}}.
\end{align*}
The uniform convergence follows along similar lines. 
\end{proof}

\section{Analysis of the generating function} \label{Section_generating matrix}

The focus of this section is on examining the generating function for the number of real eigenvalues.

\subsection{Proof of Proposition~\ref{Prop_finite N}}  \label{Subsection_generating function}

We first prove Proposition~\ref{Prop_finite N}.

\begin{proof}[Proof of Proposition~\ref{Prop_finite N}]
Let us write 
\begin{equation} \label{inner product}
\langle f,g \rangle := \langle f,g \rangle_\R+\langle f,g \rangle_\C,
\end{equation}
where
\begin{align}
&\langle f,g \rangle_\R := \frac12 \int_{\R^2} \,dx \,dy \, e^{-\frac{x^2+y^2}{2(1+\tau)} } \sgn(y-x) f(x) g(y),
\\
&\langle f,g \rangle_\C := i \int_{\R} \,dx \int_0^\infty \,dy \, e^{\frac{y^2-x^2}{1+\tau} } \erfc\Big( \sqrt{ \frac{2}{1-\tau^2} y}\Big) [f(x+iy)g(x-iy)- g(x+iy)f(x-iy)].
\end{align}
In terms of the scaled monic Hermite polynomials 
\begin{equation} \label{Ck scaled Hermite}
C_k(z):= \Big( \frac{\tau}{2} \Big)^{k/2} H_k\Big( \frac{z}{ \sqrt{2\tau} } \Big),
\end{equation}
we define
\begin{equation} \label{SOP}
q_{2j}(x):=C_{2j}(x), \qquad q_{2j+1}(x):=C_{2j+1}(x)-2j C_{2j-1}(x). 
\end{equation}
Then by \cite[Theorem 1]{MR2430570}, $\{q_j\}$ forms a family of monic skew-orthogonal polynomials with respect to \eqref{inner product}. 
Here, the skew-norm $r_j$ is given by  
\begin{equation} \label{rj skew norm}
r_j:= \langle  q_{2j},q_{2j+1}  \rangle =\sqrt{2\pi} \,(1+\tau) \Gamma(2j+1). 
\end{equation}
We also write 
\begin{equation} \label{A skew norm real}
\AA_{j,k}=\langle q_{j-1},q_{k-1} \rangle_\R.  
\end{equation}

Let 
\begin{equation} \label{def of g}
g(z) \equiv g^{(\tau)}_{2n}(z):=\sum_{k=0}^{n} z^k p_{2n,2k}^{(\tau)} . 
\end{equation}
Then along the lines of the proof of \cite[Lemma 2.1]{MR3563192}, it follows that 
\begin{equation}
\begin{split} \label{g 2n det Ajk}
g^{(\tau)}_{2n}(z) 
= \det\Big[ \delta_{jk}+\frac{z-1 }{ \sqrt{r_{j-1} r_{k-1} } }\textbf{A}_{2j-1,2k} \Big]_{j,k=1}^n. 
\end{split}
\end{equation}

Therefore it suffices to evaluate \eqref{A skew norm real}. 
For this purpose, let
\begin{equation}
\begin{split}
I_{j,k} & = \int_{ \R^2 } \,dx\,dy\, e^{ -\frac{ x^2+y^2 }{ 2(1+\tau) } } C_{2j+1}(x) C_{2k}(y)\sgn(y-x).
\end{split}
\end{equation}
Then by \eqref{SOP}, $\AA_{2j-1,2k}$ can be written in terms of $I_{j,k}$ as 
\begin{align} \label{A I jk}
\AA_{2j-1,2k} 
&= -\frac12 \Big( I_{k-1,j-1}-2(k-1)I_{ k-2,j-1} \Big).
\end{align}
On the other hand, by \cite[Eq.(5.9)]{MR2430570}, $I_{j,k}$ satisfies the recurrence relation
\begin{equation} \label{I jk recur}
I_{j+1,k}=(2j+2)I_{j,k}-2 \xi_{j,k}, \qquad \xi_{j,k}=(1+\tau) \int_\R e^{ -\frac{x^2}{1+\tau} } C_{2j+2}(x)C_{2k}(x) \,dx. 
\end{equation}
Combining \eqref{A I jk} and \eqref{I jk recur}, we have
\begin{equation} \label{A xi jk}
\AA_{2j-1,2k}  = \xi_{ k-2,j-1 }. 
\end{equation}

Now it remains to evaluate $I_{j,k}.$ 
For this, we use the following integration formula that can be found in the proof of \cite[Lemma 5.2]{byun2021real}: for $j+k$ even, 
\begin{align}
	\begin{split}
&\quad \int_\R e^{-\frac{x^2}{1+\tau}} H_j\Big( \frac{x}{\sqrt{2\tau}} \Big)H_k\Big( \frac{x}{\sqrt{2\tau}} \Big) \,dx
\\
&=\Big( \frac{1+\tau}{1-\tau} \Big)^{\frac12} \Big( \frac{\tau}{2} \Big)^{ -\frac{j+k}{2} } \Gamma\Big(\frac{j+k+1}{2}\Big) {}_2F_1\Big(\frac{j-k+1}{2},\frac{k-j+1}{2};\frac{1-j-k}{2};-\frac{\tau}{1-\tau}\Big).
	\end{split}
\end{align}
Then by \eqref{Ck scaled Hermite}, we have
\begin{equation} \label{xi jk}
\begin{split}
\xi_{j,k} & =(1+\tau) \Big( \frac{\tau}{2} \Big)^{j+k+1} \int_\R e^{ -\frac{x^2}{1+\tau} } H_{2j+2}\Big( \frac{x}{ \sqrt{2\tau} } \Big)H_{2k}\Big( \frac{x}{ \sqrt{2\tau} } \Big)  \,dx
\\
&= (1+\tau) \Big( \frac{1+\tau}{1-\tau} \Big)^{\frac12} \Gamma\Big(j+k+\frac{3}{2}\Big) {}_2F_1\Big(j-k+\frac{3}{2},k-j-\frac{1}{2};-j-k-\frac12;-\frac{\tau}{1-\tau}\Big).
\end{split}
\end{equation}
Then it follows from \eqref{A xi jk} and \eqref{xi jk} that 
\begin{equation}
\begin{split} \label{A jk evaluation}
\AA_{2j-1,2k}  &=  (1+\tau) \Big( \frac{\tau}{2} \Big)^{j+k-2} \int_\R e^{ -\frac{x^2}{1+\tau} } H_{2j-2}\Big( \frac{x}{ \sqrt{2\tau} } \Big) H_{2k-2}\Big( \frac{x}{ \sqrt{2\tau} } \Big)   \,dx
\\
&= (1+\tau)   \Big( \frac{1+\tau}{1-\tau} \Big)^{\frac12}  \Gamma\Big(j+k-\frac{3}{2}\Big) {}_2F_1\Big(k-j+\frac{1}{2},j-k+\frac{1}{2};-j-k+\frac52;-\frac{\tau}{1-\tau}\Big).
\end{split}
\end{equation}
Combining \eqref{g 2n det Ajk}, \eqref{rj skew norm} and \eqref{A jk evaluation}, we obtain the desired identity \eqref{Mn tau jk}, where the second expression follows from the reflection formula of the Gamma function
\begin{equation}\label{Gamma reflection}
\Gamma(z)\Gamma(1-z)=\pi/\sin(\pi z).
\end{equation} 
\end{proof}

\subsection{Evaluations of trace powers} \label{Subsection_generating matrix evaluation}

In this subsection, we derive two different expressions of $\Tr (M_n^{(\tau)})^m$ in Lemmas~\ref{lem:TrMnmIntegral} and \ref{Lem_Tr M^m sum}.

\begin{lem} \label{lem:TrMnmIntegral}
We have  
\begin{align}
\Tr (M_n^{(\tau)})^m = 
\int_{\mathbb R^m}  K_n^{(\tau)}(x_1,x_2) K_n^{(\tau)}(x_2,x_3)\cdots K_n^{(\tau)}(x_m,x_1) \, dx_1 \cdots dx_m,
\end{align}
where
\begin{align} \label{Kn tau Hermite}
K_n^{(\tau)}(x,y) &:= \frac{1}{\sqrt{2\pi}} e^{-\frac{x^2+y^2}{2(1+\tau)}}
\sum_{j=0}^{n-1} \frac{(\tau/2)^{2j}}{(2j)!} H_{2j}\Big(\frac{x}{\sqrt{2\tau}}\Big)
H_{2j}\Big(\frac{y}{\sqrt{2\tau}}\Big).
\end{align}
\end{lem}
\begin{proof}
By definition, 
\begin{equation}
\Tr (M_n^{(\tau)})^m  = \sum_{ j_1, j_2, \dots, j_m=1 }^n M_n^{(\tau)}(j_1,j_2) M_n^{(\tau)}(j_2,j_3) \cdots M_n^{(\tau)}(j_{m-1},j_m) M_n^{(\tau)}(j_m,j_1).
\end{equation}
Then the expression \eqref{Kn tau Hermite} follows by plugging in the first line of \eqref{Mn tau jk}.
\end{proof}

\begin{rem}
The kernel $K_n^{(\tau)}$ can also be written in terms of the Laguerre 
polynomials 
\begin{equation}
	L_j^\nu(z):= \sum_{k=0}^j \frac{\Gamma(j+\nu+1)}{(j-k)!\,\Gamma(\nu+k+1)} \frac{(-z)^k}{k!} 
\end{equation} 
as 
\begin{equation}
K_n^{(\tau)}(x,y)= \frac{ 1 }{ \sqrt{2} } e^{-\frac{x^2+y^2}{2(1+\tau)}}
\sum_{j=0}^{n-1}  \frac{\tau^{2j}\, j! }{ \Gamma(j+\frac12) } L_{j}^{-1/2}\Big(\frac{x^2}{2\tau}\Big)
L_{j}^{-1/2}\Big(\frac{y^2}{2\tau}\Big).   \label{Kn tau Laguerre}
\end{equation}
This follows from the relation
\begin{equation} \label{Hermite Laguerre}
H_{2n}(x)= (-1)^n 2^{2n} n! L_n^{-1/2} (x^2)
\end{equation}
and the duplication formula of the gamma function:
\begin{equation} \label{Gamma duplication}
\Gamma(2z+1) = \frac{1}{\sqrt{\pi}} 2^{2z} \Gamma(z+\frac12) \Gamma(z+1). 
\end{equation}
\end{rem}

\begin{rem} \label{rem:operatorInterpretation}
Let $T_n^{(\tau)}$ be defined as the operator
\begin{align*}
    f &\mapsto T_n^{(\tau)}(f):= \int_{\mathbb R} f(y) K_n^{(\tau)}(x, y) dy. 
\end{align*}
Then we have
\begin{align*}
    -\log p_{2n,0}^{(\tau)}  &= \sum_{m=1}^\infty \frac{1}{m} \Tr (M_n^{(\tau)})^m= \sum_{m=1}^\infty \frac{1}{m} \int_{\mathbb R^m} K_n^{(\tau)}(x_1, x_2) \cdots K_n^{(\tau)}(x_m, x_1) \prod_{k=1}^m \,dx_k
    \\
    &= \sum_{m=1}^\infty \frac{1}{m} \Tr\left((T_n^{(\tau)})^m\right) = -\log \det(1-T_n^{(\tau)}),
\end{align*}
where the determinant in the last line is the Fredholm determinant,
\begin{align*}
    \det(1-T_n^{(\tau)}) = 1+ \sum_{m=1}^\infty \frac{(-1)^m}{m!} \int_{\mathbb R^m} \det \left(K_n^{(\tau)}(x_j, x_k)\right)_{1\leq j, k\leq m} \prod_{k=1}^m \,dx_k.
\end{align*}
\end{rem}

\begin{ex}
Note that by \eqref{Hermite to monomial}, 
\begin{equation}
K_n^{(0)}(x,y)= \frac{ e^{-\frac{x^2+y^2}{2} }  }{ \sqrt{2\pi} } \sum_{j=0}^{n-1} \frac{ (xy)^{2j} }{ (2j)! } =  \frac{ e^{-\frac{x^2+y^2}{2} }  }{ \sqrt{2\pi} } \cosh_{n-1}(xy), 
\end{equation}
where $\cosh_n(x)=\sum_{j=0}^n x^{2j}/(2j)!.$
Therefore 
\begin{align*}
&\quad \Tr (M_n^{(\tau)})^m  =
\int_{\mathbb R^m}  \frac{ e^{-x_1^2-x_2^2-\dots -x_m^2} }{(2\pi)^{m/2}} \cosh_{n-1}(x_1x_2) \cosh_{n-1}(x_2x_3) \dots \cosh_{n-1}(x_mx_1) \, dx_1 \cdots dx_m
\\
&= \int_0^\infty \frac{ dx_1 }{ \sqrt{2\pi x_1} }  \int_0^\infty \frac{ dx_2 }{ \sqrt{2\pi x_2} }  \dots  \int_0^\infty \frac{ dx_m }{ \sqrt{2\pi x_m} } e^{-x_1-\dots-x_m} \cosh_{n-1}(\sqrt{x_1x_2})\dots \cosh_{n-1}(\sqrt{x_mx_1}),
\end{align*}
which corresponds to the integral representation in \cite[Eq.(A.23)]{MR3563192}. 
\end{ex}

Next, we show the following. 

\begin{lem} \label{Lem_Tr M^m sum}
Let $j_1,j_2,\dots, j_m=j_0 \in \mathbb{N}.$ Then for any $n,m \in \mathbb{N}$, we have  
\begin{equation} \label{Tr M^m sum evaluation}
\Tr (M_n^{(\tau)})^m =   \sum_{j_1,\dots,j_m=0}^{n-1}  \Big(\frac{ 1+\tau }{ 2 }\Big)^{\frac{m}{2}+ 2\sum_{k=1}^m j_k  } \prod_{k=1}^m  \bigg( \sum_{ l=0 }^{j_{k-1}}    \frac{ (1-\tau)^{2l }(1+\tau)^{-2l}  \, (2j_k)!  }{ 2^{2l} l! (l+j_{k}-j_{k-1})! (2j_{k-1}-2l)!  } \bigg) .
\end{equation}
\end{lem}

\begin{rem} \label{eq:suml=0j1=...=jm}
For given $l$ in the inner summation, it suffices to consider the summands with 
\begin{equation}
 j_k \ge l, \qquad j_k \ge j_{k-1}-l ,\qquad k=1,2,\dots, m.
\end{equation}
In particular, if $l=0,$ this gives $j_1=j_2=\dots=j_m.$
\end{rem}

\begin{rem}
Recall that for $\tau=1$, we have $M_n^{(1)}=I.$ 
Thus $\Tr (M_n^{(1)})^m  = n. $
This can be checked using the identity \eqref{Tr M^m sum evaluation}; namely, for $\tau=1$, it reads
\begin{align} \label{Tr M^m sum evaluation tau1} 
\Tr (M_n^{(1)})^m &=  \sum_{j_1,\dots,j_m=0}^{n-1}  \prod_{k=1}^m      \frac{  (2j_k)!    }{  (j_{k}-j_{k-1})! (2j_{k-1})!  } =   \sum_{j_1,\dots,j_m=0}^{n-1}  \prod_{k=1}^m      \frac{  1   }{  (j_{k}-j_{k-1})!  } =  \sum_{j=0}^{n-1}  \prod_{k=1}^n \frac{1}{0!}  =n.
\end{align}
\end{rem}

\begin{proof}[Proof of Lemma~\ref{Lem_Tr M^m sum}]
Using the contour integral representation 
\begin{equation}
H_k(x)= \frac{k!}{2\pi i} \oint \frac{ e^{ 2\zeta x-\zeta^2 } }{  \zeta^{k+1} }\,d\zeta,  
\end{equation}
and \eqref{Kn tau Hermite}, we have
\begin{align}
 K_n^{(\tau)}(x,y) =  \frac{1}{\sqrt{2\pi}} e^{-\frac{x^2+y^2}{2(1+\tau)}}
\sum_{j=0}^{n-1} (2j)! \Big(\frac{\tau}{2}\Big)^{2j}  \underset{\zeta,\eta=0}{\textup{Res}}  \bigg[ \frac{ e^{ 2\zeta \frac{x}{\sqrt{2\tau}} +2\eta \frac{y}{\sqrt{2\tau}} -\zeta^2-\eta^2 } }{  (\zeta \eta)^{2j+1} } \bigg].
\end{align}
Here and in the sequel, we shall use the shorthand notation 
$$
\underset{\zeta_1,\dots, \zeta_k=0}{\text{Res}}[f(\zeta,\eta)]:=\underset{\zeta_1=0}{\text{Res}}[ \cdots  [ \underset{\zeta_k=0}{\text{Res}}[f(\zeta_1,\dots,\zeta_k)]].
$$
Thus we obtain 
\begin{align*}
&\quad  K_n^{(\tau)}(x_1,x_2) K_n^{(\tau)}(x_2,x_3)\cdots K_n^{(\tau)}(x_m,x_1)
\\
&= \frac{1}{(2\pi)^{m/2}} e^{ -\frac{x_1^2+\dots x_m^2 }{ 1+\tau }  } \sum_{j_1,\dots,j_m=0}^{n-1} (2j_1)! \dots (2j_m)! \Big(\frac{\tau}{2}\Big)^{2(j_1+\dots+j_m)}
\\
&\quad \times   \underset{ \substack{\zeta_{k},\eta_{k}=0 \\ k=1,\dots,m} }{\textup{Res}}  \bigg[   \prod_{k=1}^m \exp\Big( \sqrt{ \frac{2}{\tau} } (\zeta_{k}+\eta_{k}) x_k -\zeta_{k}^2-\eta_{k}^2 \Big)  \frac{ 1 }{  \zeta_{k}^{2j_{k-1}+1} \eta_{k}^{ 2j_k+1 }   } \bigg].
\end{align*}
Since
\begin{equation}
\int_\R e^{ -\frac{ x^2 }{ 1+\tau } +\sqrt{ \frac{2}{\tau} } (\zeta+\eta) x } \,dx= \sqrt{\pi(1+\tau)} \exp\Big( \frac{1+\tau}{2\tau} (\zeta+\eta)^2 \Big),
\end{equation}
it follows from Lemma~\ref{lem:TrMnmIntegral} that
\begin{align}
\begin{split}
 \Tr (M_n^{(\tau)})^m 
&= \Big( \frac{1+\tau}{2} \Big)^{m/2}  \sum_{j_1,\dots,j_m=0}^{n-1} (2j_1)! \dots (2j_m)! (\tau/2)^{2(j_1+\dots+j_m)}
\\
&\quad \times  \underset{ \substack{\zeta_{k},\eta_{k}=0 \\ k=1,\dots,m} }{\textup{Res}}   \bigg[   \prod_{k=1}^m \exp\Big( \frac{1+\tau}{2\tau} (\zeta_{k}+\eta_{k})^2-\zeta_{k}^2-\eta_{k}^2 \Big)   \frac{ 1 }{  \zeta_{k}^{2j_{k-1}+1} \eta_{k}^{ 2j_k+1 }   } \bigg].
\end{split}
\end{align}

Note that 
\begin{align*}
&\quad   \underset{ \substack{\zeta_{k},\eta_{k}=0 \\ k=1,\dots,m} }{\textup{Res}}   \bigg[   \prod_{k=1}^m \exp\Big( \frac{1+\tau}{2\tau} (\zeta_{k}+\eta_{k})^2-\zeta_{k}^2-\eta_{k}^2 \Big)   \frac{ 1 }{  \zeta_{k}^{2j_{k-1}+1} \eta_{k}^{ 2j_k+1 }   } \bigg]
\\
&=\prod_{k=1}^m \underset{\zeta_{k},\eta_{k}=0}{\textup{Res}}  \bigg[    \exp\Big( \frac{1+\tau}{2\tau} (\zeta_{k}+\eta_{k})^2-\zeta_{k}^2-\eta_{k}^2 \Big)   \frac{ 1 }{  \zeta_{k}^{2j_{k-1}+1} \eta_{k}^{ 2j_k+1 }   } \bigg].
\end{align*}
Since 
\begin{align*}
 \exp\Big( \frac{1+\tau}{2\tau} (\zeta+\eta)^2-\zeta^2-\eta^2 \Big) &= \sum_{k=0}^\infty \frac{1}{k!} \Big( \frac{1-\tau}{2\tau} (\zeta^2+\eta^2)+\frac{1+\tau}{2\tau}  2 \zeta \eta \Big)^k 
\end{align*}
we have, applying Newton's binomial formula twice, that (only $k=p+q$ can survive)  
\begin{align*}
&\quad \underset{\zeta,\eta=0}{\textup{Res}}  \bigg[    \exp\Big( \frac{1+\tau}{2\tau} (\zeta+\eta)^2-\zeta^2-\eta^2 \Big)   \frac{ 1 }{  \zeta^{2p+1} \eta^{2q+1}  } \bigg]
\\
&= \frac{1}{(p+q)!} \sum_{l=0}^{p} \Big( \frac{1-\tau}{2\tau} \Big)^{2l+q-p} \Big( \frac{1+\tau}{2\tau} \Big)^{2p-2l} 2^{2p-2l} \binom{p+q}{2l-p+q} \binom{2l-p+q}{l}\\
&= \frac{1}{(p+q)!} \sum_{ l=0 }^p   \Big( \frac{1-\tau}{2\tau} \Big)^{2l+q-p} \Big( \frac{1+\tau}{2\tau} \Big)^{2p-2l} \frac{2^{2p-2l}  (p+q)!  }{ l! (l+q-p)! (2p-2l)!  }
\\
&= \frac{ 2^{p-q}  }{\tau^{p+q}  } \sum_{ l=0 }^p    \frac{ (1-\tau)^{2l+q-p}(1+\tau)^{2p-2l}    }{ 2^{2l} l! (l+q-p)! (2p-2l)!  } =: f(p,q).
\end{align*}
This gives 
\begin{align*}
\underset{ \substack{\zeta_{k},\eta_{k}=0 \\ k=1,\dots,m} }{\textup{Res}}   \bigg[   \prod_{k=1}^m \exp\Big( \frac{1+\tau}{2\tau} (\zeta_{k}+\eta_{k})^2-\zeta_{k}^2-\eta_{k}^2 \Big)   \frac{ 1 }{  \zeta_{k}^{2j_{k-1}+1} \eta_{k}^{ 2j_k+1 }   } \bigg]=\prod_{k=1}^m f(j_{k-1}, j_{k} )
\end{align*}
Therefore we obtain 
\begin{align*}
\begin{split}
&\quad  \Tr (M_n^{(\tau)})^{m} = \Big( \frac{1+\tau}{2} \Big)^{m/2}  \sum_{j_1,\dots,j_m=0}^{n-1} (2j_1)! \dots (2j_m)! \Big(\frac{\tau}{2}\Big)^{2(j_1+\dots+j_m)} \prod_{k=1}^m f(j_{k-1}, j_{k} )
\\
&=  \Big( \frac{1+\tau}{2} \Big)^{m/2}  \sum_{j_1,\dots,j_m=0}^{n-1} (2j_1)! \dots (2j_m)! \Big(\frac{1}{2}\Big)^{2(j_1+\dots+j_m)} \prod_{k=1}^m \bigg(  \sum_{ l=0 }^{j_{k-1}}    \frac{ (1-\tau)^{2l+j_k-j_{k-1} }(1+\tau)^{2j_{k-1}-2l}    }{ 2^{2l} l! (l+j_{k}-j_{k-1})! (2j_{k-1}-2l)!  } \bigg), 
\end{split}
\end{align*}
which gives the lemma. 
\end{proof}

\subsection{Estimates of trace powers} \label{Subsection_generating matrix estimates}

\begin{lem} \label{lem:TrMnmMono}
We have
\begin{align}
    \Tr (M_n^{(\tau)})^m \leq 
    \frac{1}{(2\pi)^{\frac{m-1}{2}}} \int_{\mathbb R^m} 
    e^{-\frac{1}{2}(x_1^2+x_m^2)}
    \mathcal K_n^{(\tau)}(x_m,x_1)
    e^{-\sum_{j=2}^{m-1} x_j^2}
    \prod_{j=2}^{m} \cosh\left(x_{j-1}x_j\right)
  \, dx_1 \cdots dx_m,
\end{align}
where
\begin{align} \label{mathcal Kn tau}
    \mathcal K_n^{(\tau)}(x,y) & : = \sqrt{1-\tau^2} K_n^{(\tau)}(\sqrt{1-\tau^2}\, x, \sqrt{1-\tau^2}\, y). 
\end{align}
Here, $K_n^{(\tau)}$ is given by \eqref{Kn tau Hermite}.  
\end{lem}

\begin{proof}
Recall the multiplication theorem for Hermite polynomials 
\begin{equation}
\begin{split}
H_{2n}(\lambda x) 
= \sum_{k=0}^n \frac{ (2n)! }{ (2k)! \,(n-k)! } \lambda^{2k} (\lambda^{2}-1)^{n-k} H_{2k}(x), 
\end{split}
\end{equation}
see e.g. \cite[Eq.(18.18.13)]{olver2010nist}. 
Letting   
$$e_\tau = \frac{1+\tau}{\sqrt{1-\tau^2}},$$ 
we have
\begin{align*}
   H_{2j}\Big(\sqrt{\frac{1-\tau^2}{2\tau}}  x\Big)
  &  = H_{2j}\Big(\frac{e_\tau}{\sqrt{2\tau}}  (1-\tau) x\Big)
   = \sum_{k=0}^j \Big(\frac{e_\tau}{\sqrt{\tau}}\Big)^{2k} \Big(\frac{e_\tau^2}{\tau}-1\Big)^{j-k} \frac{(2j)!}{ (2k)! (j-k)!} H_{2k}\Big(\frac{1-\tau}{\sqrt 2} x\Big). 
\end{align*}
Note that 
\begin{align*}
&\quad \sum_{m=0}^{n-1} \frac{(\tau/2)^{2m}}{(2m)!} H_{2m}\Big(  \sqrt{ \frac{1-\tau^2}{2\tau} }  x  \Big)
H_{2m}\Big( \sqrt{ \frac{1-\tau^2}{2\tau} }  y \Big)
\\
&= \sum_{m=0}^{n-1} \sum_{j,k=0}^m \Big(\frac{1}{1-\tau e_\tau^{-2}}\Big)^{j+k}  \Big(\frac{e_\tau^2}{\tau}-1\Big)^{2m} \frac{(\tau/2)^{2m} (2m)!}{ (2k)! (m-k)! (2j)!(m-j)!} H_{2k}\Big(\frac{1-\tau}{\sqrt 2} x\Big) H_{2j}\Big(\frac{1-\tau}{\sqrt 2} y\Big).
\end{align*}
We can thus write the kernel as a combination of Hermite functions
\begin{align}
\begin{split}
    \mathcal K_n^{(\tau)}(x, y) & = \sqrt{ \frac{1-\tau^2}{2\pi}  } e^{-\frac{1-\tau}{2}(x^2+y^2) }
\sum_{j=0}^{n-1} \frac{(\tau/2)^{2j}}{(2j)!} H_{2j}\Big(  \sqrt{ \frac{1-\tau^2}{2\tau} }  x  \Big)
H_{2j}\Big( \sqrt{ \frac{1-\tau^2}{2\tau} }  y \Big)
\\
  &  = \sum_{j,k=0}^{n-1}  e^{-\frac{1-\tau}{2} x^2}
    H_{2j}\Big(\frac{1-\tau}{\sqrt 2} x\Big)
    \hat M_{n,jk} e^{-\frac{1-\tau}{2} y^2} H_{2k}\Big(\frac{1-\tau}{\sqrt 2} y\Big)
\end{split}
\end{align}
for some matrix $\hat M_n$ that has positive elements. In fact, we see that
\begin{equation}
\hat M_{n,jk} =  \sqrt{ \frac{1-\tau^2}{2\pi}  }  \Big(\frac{1}{1-\tau e_\tau^{-2}}\Big)^{j+k} \sum_{m=\max(j,k)}^{n-1} \frac{(\tau/2)^{2m} (2m)!}{ (2k)! (m-k)! (2j)!(m-j)!} 
\Big(\frac{e_\tau^2}{\tau}-1\Big)^{2m}.    
\end{equation}

By Lemma~\ref{lem:TrMnmIntegral} and the change of variables, we have 
\begin{align}
\Tr (M_n^{(\tau)})^m = 
\int_{\mathbb R^m}  \mathcal K_n^{(\tau)}(x_1,x_2) \mathcal K_n^{(\tau)}(x_2,x_3)\cdots \mathcal K_n^{(\tau)}(x_m,x_1) \, dx_1 \cdots dx_m,
\end{align}
Since the Hermite functions are orthonormal, we infer that $\Tr (M_n^{(\tau)})^m$ is a sum of products of elements of $\hat M_n$ (which are positive). 

Now suppose that we add an extra term
\begin{align*}
   \frac{(\tau/2)^{2n}}{(2n)!} H_{2n}\Big(\sqrt{\frac{1-\tau^2}{2\tau}}  x_1\Big)
    H_{2n}\Big(\sqrt{\frac{1-\tau^2}{2\tau}}  x_2\Big)
\end{align*}
to $\mathcal K_n^{(\tau)}(x_1, x_2)$. The corresponding matrix $\hat M_{n+1}$ then has elements that are greater than or equal to their counterparts of $\hat M_n$. Furthermore, it has $2n-1$ elements more, which are all positive. Hence the sum over products of elements of $\hat M_{n+1}$ and $\hat M_n$ has increased overall. That is, $\Tr (M_n^{(\tau)})^m$ has increased. We can repeat this argument inductively and extend the summation in the kernel \eqref{Kn tau Hermite} over all non-negative integers $j$. Recall here that the Mehler kernel formula is given by
\begin{equation}
\sum_{j=0}^\infty \frac{ (\tau/2)^j }{j!} H_j(x)H_j(y)= \frac{1}{ \sqrt{1-\tau^2} } \exp\Big( \frac{2\tau}{ 1-\tau^2 } x y-\frac{\tau^2}{ 1-\tau^2 }(x^2+y^2) \Big),
\end{equation}
which leads to  
\begin{align*}
 \sum_{j=0}^\infty \frac{ (\tau/2)^{2j} }{(2j)!} H_{2j}(x)H_{2j}(y) &= \frac12 \Big(\sum_{j=0}^\infty \frac{ (\tau/2)^j }{j!} H_j(x)H_j(y)+ \sum_{j=0}^\infty \frac{ (\tau/2)^j }{j!} H_j(-x)H_j(y) \Big)
\\
&= \frac{1}{ \sqrt{1-\tau^2} } e^{ -\frac{\tau^2}{ 1-\tau^2 }(x^2+y^2)  } \cosh\Big( \frac{2\tau}{1-\tau^2} xy  \Big).
\end{align*}
Using this, the infinite sum evaluates to
\begin{align*}
 \mathcal K_\infty^{(\tau)}(x, y) 
 = \frac{1}{\sqrt{2\pi}} e^{-\frac{1}{2}(x^2+y^2)} \cosh(x y).
\end{align*}
Replacing $m-1$ of the $m$ kernels by this expression, we arrive at the result. 
\end{proof}

\begin{lem}\label{Lem_TrM power ineq contour}
We have
\begin{align} \label{TrM power ineq contour}
    \Tr (M_n^{(\tau)})^m &\leq 
    \frac{\sqrt{1+\tau}}{2\pi} \oint_\gamma \frac{s^{-2n+1}}{\sqrt{(1-s)((m-1-(m+1)\tau)s+(m+1)-(m-1)\tau)}} \frac{ds}{1-s^2},
\end{align}   
where $\gamma$ is a small loop around $0$ with positive direction.
\end{lem}

\begin{rem}
For the case $\tau=0$, the inequality \eqref{TrM power ineq contour} agrees with \cite[Eq.(A.31)]{MR3563192}.
\end{rem}

\begin{proof}[Proof of Lemma~\ref{Lem_TrM power ineq contour}]
We start with the inequality from Lemma \ref{lem:TrMnmMono}.
    We use Lemma \ref{lem:eCoshProdInt2} for the integrations over $x_2,\ldots,x_{m-1}$. This yields
    \begin{align} \nonumber
    \Tr (M_n^{(\tau)})^m &\leq 
    \frac{1}{\sqrt{2\pi(m-1)}} \int_{\mathbb R^2}     
    \mathcal K_n^{(\tau)}(x_m,x_1)
    e^{-\frac{1}{2(m-1)}(x_1^2+x_m^2)}
     \cosh\Big(\frac{x_1 x_m}{m-1}\Big)
  \, dx_1 dx_m\\ \label{eq:lem:TrMnmMono1}
  &= \frac{1}{\sqrt{2\pi(m-1)}} \int_{\mathbb R^2}     
    \mathcal K_n^{(\tau)}(x_m,x_1)
    e^{-\frac{1}{2(m-1)}(x_1-x_m)^2}
       \, dx_1 dx_m.
\end{align}
Here, we have used the symmetry $K_n^{(\tau)}(x_m,x_1)= K_n^{(\tau)}(-x_m,x_1)$. 
    Then we plug in the single integral representation for the kernel from \cite[Eq.(26)]{ADM}, which, adapted to only yield even indexed terms, takes the form
    \begin{equation} \label{eq:ADMintegralRep}
       \mathcal K_n^{(\tau)}(x,y) = \sqrt{\frac{1-\tau^2}{2\pi}} e^{-\frac{1-\tau}{2}(x^2+y^2)} \frac{1}{2\pi i} \oint_\gamma e^{\frac{1-\tau^2}{4\tau} \frac{s}{1+s} (x+y)^2-\frac{1-\tau^2}{4\tau} \frac{s}{1-s} (x-y)^2} \frac{(\tau/s)^{2n}}{(\tau/s)^2-1} \frac{ds}{s\sqrt{1-s^2}}.
    \end{equation}
Interchanging the order of integration, we have 
   \begin{align*}
    &\quad \frac{1}{\sqrt{2\pi(m-1)}} \int_{\mathbb R^2}     
    \mathcal K_n^{(\tau)}(x_m,x_1)
    e^{-\frac{1}{2(m-1)}(x_1-x_m)^2}
       \, dx_1 \, dx_m
       \\
       &= \sqrt{ \frac{1-\tau^2}{ m-1 } } \frac{1}{4\pi^2 i}  \oint_\gamma  \int_{ \R^2 }  e^{-\frac{1-\tau}{2}(x^2+y^2) + \frac{1-\tau^2}{2\tau} \frac{s}{1+s} x^2-(\frac{1-\tau^2}{2\tau} \frac{s}{1-s} +\frac{1}{m-1}  ) y^2  }    \, dx \, dy  \frac{(\tau/s)^{2n}}{(\tau/s)^2-1} \frac{ds}{s\sqrt{1-s^2}},
   \end{align*}
   where now we use the light-cone coordinates:
   $$
   x=\frac{x_1+x_m}{\sqrt{2}}, \qquad y= \frac{x_1-x_m}{\sqrt{2}}.
   $$
   Since 
   \begin{align*}
    &\quad \int_{ \R^2 }  e^{-\frac{1-\tau}{2}(x^2+y^2) + \frac{1-\tau^2}{2\tau} \frac{s}{1+s} x^2-(\frac{1-\tau^2}{2\tau} \frac{s}{1-s} +\frac{1}{m-1}  ) y^2  }    \, dx \, dy 
    \\
    &=\pi  \Big( -\frac{1-\tau}{2}  + \frac{1-\tau^2}{2\tau} \frac{s}{1+s} \Big)^{-\frac12} \Big( -\frac{1-\tau}{2} -\frac{1-\tau^2}{2\tau} \frac{s}{1-s} -\frac{1}{m-1}   \Big)^{-\frac12}
    \\
    &= 2\pi \sqrt{ \frac{1-s^2}{ 1-\tau } (m-1) } \Big( \frac{1}{ s/\tau-1 } \frac{1}{  (1+m+\tau-m\tau)-s(1-m+\tau m +\tau)/\tau   } \Big)^{\frac12},
   \end{align*}
   we have 
   \begin{align*}
    &\quad \sqrt{ \frac{1-\tau^2}{ m-1 } } \frac{1}{4\pi^2 i}   \int_{ \R^2 }  e^{-\frac{1-\tau}{2}(x^2+y^2) + \frac{1-\tau^2}{2\tau} \frac{s}{1+s} x^2-(\frac{1-\tau^2}{2\tau} \frac{s}{1-s} +\frac{1}{m-1}  ) y^2  }    \, dx \, dy    \frac{(\tau/s)^{2n}}{(\tau/s)^2-1} \frac{1}{s\sqrt{1-s^2}}
    \\
    &= \frac{ \sqrt{1+\tau} }{2\pi i}      \Big( \frac{1}{ s/\tau-1 } \frac{1}{  m+1-\tau(m-1) -s( m +1 - (m-1)\tau^{-1})   } \Big)^{\frac12} \frac{(\tau/s)^{2n}s}{\tau^2-s^2}. 
   \end{align*}
    This yields 
    \begin{align*}
   \Tr (M_n^{(\tau)})^m & \leq  \frac{\sqrt{1+\tau}}{2\pi} \oint_\gamma \frac{1}{\sqrt{(1-s/\tau)(m+1-\tau(m-1)-(m+1-\tau^{-1}(m-1))s)}} \frac{(\tau/s)^{2n} s}{\tau^2-s^2} \, ds.
    \end{align*}
  We arrive at the result after a substitution $s\to \tau s$.
\end{proof}

\subsection{Proofs of Lemmas~\ref{lem:inequalityTrMNtaum} and \ref{Lem_eigenvalue of Mn tau boundedness}}

\label{Subsection_error estimates lemmas}




We now prove Lemma~\ref{lem:inequalityTrMNtaum}. 
Our proof is similar to the argument given in \cite[Appendix A.3]{MR3563192}.

\begin{proof}[Proof of Lemma~\ref{lem:inequalityTrMNtaum}]
Our starting point is the estimate \eqref{TrM power ineq contour}. Recall that $\gamma$ in \eqref{TrM power ineq contour} is a small loop around $0$ with positive direction. The integrand on the RHS in \eqref{TrM power ineq contour} has three singularities, $1, -1$ and a singularity that we will denote by
\begin{align*}
    a=-\frac{m+1-(m-1)\tau}{m-1-(m+1)\tau}.
\end{align*}
    First, we consider the case that $\frac{m-1}{m+1}>\tau$. In that case $a<-1$. 
    Then we deform $\gamma$ to a band $L_-$ around $(-\infty, a)$, a band $L_+$ around $(1,\infty)$, connected by two (almost) semicircles (and a small circle around the pole $s=-1$). The integrals over the semicircles tend to $0$ as we increase the radius to $\infty$. The residue at $s=-1$ gives a contribution $1/4$. 
    The integral over $L_-$ gives (without the prefactor)
    \begin{align*}
   &\quad  \left|2 \int_{-\infty}^{a} \frac{s^{-2n+1}}{\sqrt{(1-s)((m-1-(m+1)\tau)s+(m+1)-(m-1)\tau)}} \frac{ds}{s^2-1}\right|
   \\
   &     \leq \frac{2 |a|}{(1+|a|)^{3/2}(|a|-1)} \frac{1}{\sqrt{m-1-(m+1)\tau}} \left|\int_{-\infty}^{a} \frac{s^{-2n}}{\sqrt{s-a}} \, ds\right|
   \\
    &    = \frac{2 |a|^{-2n}}{|a|-1} \left(\frac{|a|}{1+|a|}\right)^{3/2} \frac{1}{\sqrt{m-1-(m+1)\tau}} \sqrt{\frac{\pi}{n}} \int_0^\infty \frac{(1+\frac{s}{2n})^{-2n}}{\sqrt{\pi s}} \, ds.
    \end{align*}
 We shall use the following elementary inequality 
    \begin{align} \label{Eq. A33}
        \int_0^\infty \frac{(1+\frac{s}{2n})^{-2n}}{\sqrt{\pi s}} ds \leq 1+1/n,
    \end{align}
    which can be found in \cite[Eq.(A.33)]{MR3563192}.  
  Using this, have
  \begin{align*}
    &\quad  \frac{\sqrt{1+\tau}}{2\pi} \int_{L_-} \frac{s^{-2n+1}}{\sqrt{(1-s)((m-1-(m+1)\tau)s+(m+1)-(m-1)\tau)}} \frac{ds}{1-s^2}
    \\
    &    \leq \sqrt{\frac{1+\tau}{2\pi n}} \frac{1}{|a|-1} \frac{1+1/n}{\sqrt{m-1-(m+1)\tau}}
        = \frac{1}{\sqrt{2\pi n(1+\tau)}}  \sqrt{m-1-(m+1)\tau} \frac{1+1/n}{2(1+\tau)}.
  \end{align*}
    Lastly, we estimate the integral over $L_+$. Since the integrand has a factor $(1-s)^{-3/2}$ we first perform a partial integration. 
    Then we may take the bandwidth to $0$. The dominant part (without prefactor) is given by
    \begin{align*}
   &\quad   4n \int_1^\infty \frac{s^{-2n}}{\sqrt{(1-s)((m-1-(m+1)\tau)s+(m+1)-(m-1)\tau)}} \frac{ds}{1+s}
   \\
   &     \leq \frac{4n}{2\sqrt{2m(1-\tau)}} \int_1^\infty \frac{s^{-2n}}{\sqrt{s-1}} \, ds
        = \frac{2 \sqrt n}{\sqrt{2m(1-\tau)}} \sqrt \pi \int_0^\infty \frac{(1+\frac{s}{2n})^{-2n}}{\sqrt{\pi s}} \, ds.
    \end{align*}
  Then we can again use the estimate \eqref{Eq. A33}.
    The part that is not dominant can be estimated using
    \begin{align*}
        \int_1^\infty \frac{s^{-2n}}{\sqrt{(1-s)(s-a)}} \left(\frac{1}{s-a}+\frac{2}{(s+1)^2}\right) ds
        \leq \frac{1}{\sqrt{1-a}} \int_1^\infty \frac{s^{-2n}}{\sqrt{1-s}} ds \leq \sqrt{\frac{\pi}{1-a}} (1+1/n).
    \end{align*}
   Combining all of the above, we conclude that
    \begin{multline*}
        \Tr (M_n^{(\tau)})^m 
        \leq \frac{1}{4} + \sqrt{\frac{1+\tau}{1-\tau}} \sqrt{\frac{n}{\pi m}} (1+1/n)
        + \frac{1}{2} \sqrt{\frac{1+\tau}{1-\tau}} \frac{1+1/n}{\sqrt{2\pi m n}}
        + \frac{1}{8} \sqrt{\frac{m-1-(m+1)\tau}{(1+\tau)\pi n}} (1+1/n) .
    \end{multline*}
    This is also true when $m=\frac{1+\tau}{1-\tau}$. This case is similar, but easier, since there is no band $L_-$ in this case. Some easy inequalities finish the proof.   
\end{proof}

\begin{proof}[Proof of Lemma \ref{Lem_eigenvalue of Mn tau boundedness}]

Let $v=(v_1,\dots,v_n) \in \R^n \setminus \{0\}.$ Then by \eqref{Mn tau jk}, it follows that  
\begin{align*}
\langle v, M_n^{(\tau)} v \rangle &= \sum_{j,k=1}^n M_n^{(\tau)}(j,k) v_j v_k =  \frac{1}{ \sqrt{2\pi} } \int_\R e^{ -\frac{x^2}{1+\tau} } \bigg[   \sum_{j=1}^n  \frac{  ( \tau/2 )^{j-1} }{ \sqrt{ \Gamma(2j-1) } } H_{2j-2}\Big( \frac{x}{ \sqrt{2\tau} } \Big)\,v_j \bigg]^2 \,dx  > 0,
\end{align*}
which gives rise to the first assertion. 

Next, we show the second assertion.    Let $m$ be the smallest integer such that
    \begin{align*}
        m>\frac{1+\tau}{1-\tau} n.
    \end{align*}
    Then we have for $n$ large enough
    \begin{align*}
        \frac{1}{8} \frac{1-\tau}{1+\tau} \sqrt{\frac{m}{\pi n}} (1+1/n) &\leq \frac{1}{8} \frac{1-\tau}{1+\tau} \sqrt{\frac{\frac{1+\tau}{1-\tau}n+1}{\pi n}} (1+1/n)\\
        &\leq \frac{1}{8\sqrt{\pi}} \sqrt{\frac{1-\tau}{1+\tau}} \sqrt{1+\frac{1-\tau}{1+\tau}\frac{1}{n}} (1+2/n)
        \leq \frac{1}{8\sqrt\pi}.
    \end{align*}
    Hence, by Lemma \ref{lem:inequalityTrMNtaum}, we have
    \begin{align*}
        \Tr (M_n^{(\tau)})^m \leq \frac{1}{4}+ \frac{1+1/n}{\sqrt \pi} + \frac{1}{8\sqrt \pi} \leq 1-a
    \end{align*}
    for some constant $a\in(0,1)$ when $n$ is big enough. The rest of the proof is similar (but with a different power) to \cite{MR3563192} yielding
    \begin{align*}
        \lambda_{\text{max}} \leq (1-a)^{1/m}\leq 1-\sqrt{\frac{1-\tau}{1+\tau}} \frac{a}{n}.
    \end{align*}
    The proof for weak non-Hermiticity is analogous, here one takes $m$ to be a multiple of $n^2$ (indeed, then $m\geq \frac{1+\tau}{1-\tau}$). 
\end{proof}

\section{Asymptotic analysis at strong non-Hermiticity} \label{Section_asymptotic strong}

\subsection{Asymptotics of the kernel}

Since the Hermite polynomials are odd (resp., even) for $j$ odd (resp., even), by symmetry, one can rewrite \eqref{Kn tau Hermite} as 
\begin{align}
\begin{split} \label{Kn tau Hermite decomp}
  K_n^{(\tau)}(x,y) &= \frac{ 1 }{ 2\sqrt{2\pi} } e^{-\frac{x^2+y^2}{2(1+\tau)}}\sum_{j=0}^{2n-1} \frac{(\tau/2)^j}{j!} H_j\Big(\frac{x}{\sqrt{2\tau}}\Big)
H_j\Big(\frac{y}{\sqrt{2\tau}}\Big)
\\
&\quad + \frac{ 1 }{ 2\sqrt{2\pi} } e^{-\frac{x^2+y^2}{2(1+\tau)}}\sum_{j=0}^{2n-1} \frac{(\tau/2)^j}{j!} H_j\Big(\frac{-x}{\sqrt{2\tau}}\Big)
H_j\Big(\frac{y}{\sqrt{2\tau}}\Big).
\end{split}
\end{align}
Asymptotics for strong non-Hermiticity can be directly extracted from \cite{ADM,Mo}. 
It will be convenient to define the rescaled kernel
\begin{align} \label{eq:defHatKn}
    \widehat{\mathcal K}_n^{(\tau)}(x,y)
    = \sqrt{2n} \, \mathcal K_n^{(\tau)}(\sqrt{2n} x, \sqrt{2n} y).
\end{align}
where we recall that $\mathcal{K}_n^{(\tau)}$ is given by \eqref{mathcal Kn tau}.

We also define the edge and focal point 
\begin{equation} \label{edge focal pts}
 e_\tau = \frac{1+\tau}{\sqrt{1-\tau^2}}, \qquad   f_\tau = \frac{2\sqrt\tau}{\sqrt{1-\tau^2}}.
\end{equation}
It will also be convenient to define 
\begin{equation} \label{xi function}
 \xi_x =  \begin{cases} 
    0, & |x|\leq f_\tau,\\
    \cosh^{-1}(x/f_\tau), & |x|> f_\tau.
    \end{cases}
\end{equation}

In what follows $H$ will be the Heaviside function, i.e. $H(x)=1$ for $x\geq 0$ and $H(x)=0$ for $x<0$.  

\begin{lem} \label{lem:asympKernel} Let $0<\tau<1$ and $\frac13<\mu<\frac12$ be fixed. Then there exist constants $c, C>0$ such that
  \begin{align*}
        & H(2 \xi_{e_\tau} - n^{-\mu} - \xi_x-\xi_y) \left(\sqrt\frac{n}{\pi}e^{-n(x^2+y^2)} \cosh(2nxy) - e^{-c n^{1-2\mu}}\right) - C n^\mu e^{-cn (|x|-e_\tau)^2} e^{-cn (|y|-e_\tau)^2}\\
        & \qquad \leq  \hat{\mathcal K}_n^{(\tau)}(x,y)\\
        \leq & H(2 \xi_{e_\tau} + n^{-\mu} - \xi_x-\xi_y) \left(\sqrt\frac{n}{\pi} e^{-n(x^2+y^2)} \cosh(2nxy) + e^{-c n^{1-2\mu}}\right) + C n^\mu e^{-cn (|x|-e_\tau)^2} e^{-cn (|y|-e_\tau)^2}
    \end{align*}
 uniformly for all $(x,y)\in\mathbb R^2$. 
\end{lem}

\begin{proof}
Without loss of generality we assume that $x,y\geq 0$. We shall use \cite{ADM} in what follows. This paper treats the kernel
\begin{align*}
    \kappa_n(z, w) = \frac{n}{\pi\sqrt{1-\tau^2}}\sqrt{\omega(z) \omega(w)} \sum_{j=0}^{n-1} H_j\bigg(\frac{\sqrt n \, z}{\sqrt{2\tau}}\Bigg) H_j\bigg(\frac{\sqrt n \, \overline w}{\sqrt{2\tau}}\bigg),
\end{align*}
where the weight is given by $\omega(z) = \exp\left(-n\frac{(\re z)^2}{1+\tau}-n\frac{(\im z)^2}{1-\tau}\right)$. Our kernel \eqref{eq:defHatKn} can be expressed as
\begin{align*}
    \widehat{\mathcal K}_n^{(\tau)}(x,y) = 
    \frac{\sqrt{\pi n}}{2} \sqrt{1-\tau^2} \left(\kappa_{2n}\left(\sqrt{1-\tau^2} \, x, \sqrt{1-\tau^2} \, y\right)+\kappa_{2n}\left(-\sqrt{1-\tau^2} \, x, \sqrt{1-\tau^2} \, y\right)\right).
\end{align*}
For $\xi_x+\xi_y>2\xi_{e_\tau}+n^{-\mu}$, it follows from \cite[Theorem I.1 and Remark I.2]{ADM}, and some easy estimations, that 
\begin{equation}
   \left|\widehat{\mathcal K}_n^{(\tau)}(x,y)\right| \leq C n^{\mu} e^{- n h(x)}   e^{- n h(y)},
\end{equation}
for some constant $C>0$, where $h$ is the continuous function
\begin{align}
h(x) = \begin{cases}
      \displaystyle   2\xi_{e_\tau} - \frac{(1-\tau)^2}{2\tau} x^2, & x\in[0,f_\tau],
      \smallskip 
      \\
 \displaystyle       2\xi_{e_\tau}-1 + (1-\tau)x^2 - \frac{f_\tau^2}{(x+\sqrt{x^2-f_\tau^2})^2}+ \log \Big( \frac{f_\tau^2}{(x+\sqrt{x^2-f_\tau^2})^2} \Big) , & x>f_\tau,
    \end{cases}
\end{align}
see \cite[Eq.(80)]{ADM} with $x=f_\tau \cos\eta$ or $x=f_\tau\cosh \xi$. We remark that one can alternatively use \eqref{eq:ADMintegralRep} as a starting point, and then follow the approach in \cite{ADM} to reach the same conclusion.
As shown in \cite{ADM}, $h$ has a double zero in $x=e_\tau$ and is positive for all $x\neq e_\tau$. In fact, we can show that
 \begin{align*}
     \frac{h(x)}{(x-e_\tau)^2} \geq \lim_{x\to\infty} \frac{h(x)}{(x-e_\tau)^2} = 1-\tau 
 \end{align*}
 for all $x>f_\tau$. We conclude that
 \begin{align*}
     C n^{\mu} e^{- n h(x)}
    e^{- n h(y)} \leq C n^\mu e^{-cn (x-e_\tau)^2 } e^{-cn (y-e_\tau)^2 }
 \end{align*}
 where $c>0$ is some constant.
For $\xi_x+\xi_y<2\xi_{e_\tau}-n^{-\mu}$ we have by \cite[Theorem III.5]{ADM} that
\begin{align*}
    \left|\widehat{\mathcal K}_n^{(\tau)}(x,y)-e^{-n(x^2+y^2)} \cosh(2nxy)\right| \leq e^{-cn^{1-2\mu}},
\end{align*}
where we possibly redefine the constant $c$. We used here that at least one of $x, y$ is $\leq e_\tau - a n^{-\mu}$ for some constant $a>0$. Lastly, we look at the region $|\xi_x+\xi_y-2\xi_{e_\tau}|\leq n^{-\mu}$. Then we have by \cite[Proposition V.1]{Mo} that
\begin{align*}
    \widehat{\mathcal K}_n^{(\tau)}(x,y) &= \frac{1}{4}\sqrt\frac{n}{\pi} e^{- n(x-y)^2} \erfc\left(\sqrt n (x+y-2e_\tau)\right) + 
     e^{-n (x-e_\tau)^2}
     e^{-n (y-e_\tau)^2} \mathcal O(n^{1-2\mu})\\
     & \quad + H(2\xi_{e_\tau}-\xi_x-\xi_y) \frac{1}{2}\sqrt\frac{n}{\pi} e^{-n(x+y)^2} 
     + \mathcal O(e^{-n h(x)} e^{-n h(y)}).
\end{align*}
Since the complementary error function takes values in $(0,2)$ for real arguments, we infer that
\begin{align*}
    \widehat{\mathcal K}_n^{(\tau)}(x,y) &\leq \sqrt\frac{n}{\pi} e^{-n (x^2+y^2)} \cosh(2nxy)  + C n^{1-2\mu} e^{-cn (x-e_\tau)^2} e^{-cn (y-e_\tau)^2}\\
    &\leq \sqrt\frac{n}{\pi} e^{-n (x^2+y^2)} \cosh(2nxy)  + C n^{\mu} e^{-cn (x-e_\tau)^2} e^{-cn (y-e_\tau)^2},
\end{align*}
for some constant $C>0$. For the lower bound we have trivially
\begin{align*}
    \widehat{\mathcal K}_n^{(\tau)}(x,y) & \geq - C n^\mu e^{-cn (x-e_\tau)^2} e^{-cn (y-e_\tau)^2}.
\end{align*}
Putting it all together, we obtain the result. 
\end{proof}

\subsection{Proof of Proposition \ref{prop:asympTracem}.}

\begin{proof}[Proof of Proposition \ref{prop:asympTracem}.]
In what follows, we shall make the identification $x_{m+1}=x_1$. By Lemma \ref{lem:asympKernel}, we infer that, for fixed $m$
\begin{align*}
   &\quad \int_{\mathbb R^m} \prod_{j=1}^m \widehat{\mathcal K}_n^{(\tau)}(x_j, x_{j+1}) \, d x_j
   \\ 
   &\geq 2^m \left(\frac{n}{\pi}\right)^\frac{m}{2} \int_{\mathbb R_+^m} e^{-2n\sum_{j=1}^m x_j^2} \prod_{j=1}^m H(2\xi_{e_\tau}-n^{-\mu}-\xi_{x_j}-\xi_{x}) \cosh(2n x_j x_{j+1}) \, dx_j - c_m n^{-\frac{1}{2}+\mu},
\end{align*}
for some constant $c_m>0$. There exists a constant $a>0$ such that
\begin{align*}
&\quad  2^m\int_{\mathbb R_+^m} e^{-2n\sum_{j=1}^m x_j^2} \prod_{j=1}^m H(2\xi_{e_\tau}-n^{-\mu}-\xi_{x_j}-\xi_{x}) \cosh(2n x_j x_{j+1}) \,dx_j
\\
&\geq \int_{[-e_\tau+a n^{-\mu}, e_\tau - a n^{-\mu}]^m} e^{-2n \sum_{j=1}^m x_j^2} \prod_{j=1}^m  \cosh(2n x_j x_{j+1}) \, d x_j
\\
& \geq 
\frac12 \int_{[-e_\tau+a n^{-\mu}, e_\tau - a n^{-\mu}]^m} e^{-2n \sum_{j=1}^m (x_j^2-x_j x_{j+1})} \prod_{j=1}^m\,d x_j.
\end{align*}
In the last step we took only the combination of exponentials $e^{\pm 2n x_j x_{j+1}}$ such that after substitutions $x_j\to \pm x_j$ we obtain $e^{n\sum_{j=1}^m x_j}$ in the integrand, which is half of all the combinations. Now we make a substitution $y_m = x_1+\ldots+x_m$ and $y_j = x_{j+1}-x_j$ for $j=1,\ldots,m-1$. Then we have
\begin{align*}
&\quad  \int_{[-e_\tau+a n^{-\mu}, e_\tau - a n^{-\mu}]^m} e^{-2n \sum_{j=1}^m (x_j^2+2 x_j x_{j+1})} \prod_{j=1}^m\,d x_j 
\\
&  = \frac{1}{m} \int_{-m(e_\tau-a n^{-\mu})}^{m(e_\tau-a n^{-\mu})} \int_{[-e_\tau+n^{-\mu}, e_\tau-n^{-\mu}]^{m-1}}
    e^{- n (\sum_{j=1}^{m-1} y_j^2+(\sum_{j=1}^{m-1} y_j)^2)} \prod_{j=1}^m\,d y_j .
\end{align*}
This we can write as
\begin{align*}
&\quad   (e_\tau-a n^{-\mu}) \int_{[-e_\tau+n^{-\mu}, e_\tau-n^{-\mu}]^{m-1}} \int_{-\infty}^\infty  e^{-n\lambda^2+ 2n \lambda \sum_{j=1}^{m-1} y_j}    e^{- n\sum_{j=1}^{m-1} y_j^2} \, d\lambda \prod_{j=1}^{m-1}\,d y_j
\\
&  = (e_\tau-a n^{-\mu}) \int_{-\infty}^\infty e^{-\frac{n}{2}\lambda^2} \bigg(\int_{-e_\tau+n^{-\mu}}^{e_\tau-n^{-\mu}} e^{-n x^2+ 2n\lambda x}\bigg)^{m-1} d\lambda
\\
 &   \geq (e_\tau-a n^{-\mu}) \int_{-e_\tau+\varepsilon}^{e_\tau-\varepsilon} e^{-n\lambda^2} \bigg(\int_{-e_\tau+n^{-\mu}}^{e_\tau-n^{-\mu}} e^{-n x^2+ 2n\lambda x}\bigg)^{m-1} d\lambda
\end{align*}
for any $\varepsilon>0$ (for $n$ big enough). We have
\begin{align*}
    \int_{-e_\tau+\varepsilon}^{e_\tau-\varepsilon} e^{-n\lambda^2} \bigg(\int_{-e_\tau+n^{-\mu}}^{e_\tau-n^{-\mu}} e^{-n x^2+ 2n\lambda x}\bigg)^{m-1} \, d\lambda
    & = \int_{-e_\tau+\varepsilon}^{e_\tau-\varepsilon} 
     e^{-n\lambda^2} \left(\frac{\pi}{n}\right)^\frac{m-1}{2} e^{-(m-1) n\lambda^2} (1+\mathcal O(1/n)) \, d\lambda
     \\
     &= 2(e_\tau - \varepsilon) \left(\frac{\pi}{n}\right)^\frac{m}{2} \frac{1}{\sqrt{m}} (1+\mathcal O(1/n)). 
\end{align*}
We conclude that
\begin{align*}
    \lim_{n\to\infty} \frac{1}{\sqrt{2n}} \Tr (M_n^{(\tau)})^m \geq (e_\tau-\varepsilon) \frac{1}{\sqrt{2\pi m}}.
\end{align*}
Since this is true for arbitrary $\varepsilon>0$, we have
\begin{align*}
    \lim_{n\to\infty}  \frac{1}{\sqrt{2n}} \Tr (M_n^{(\tau)})^m \geq e_\tau \frac{1}{\sqrt{2\pi m}} = \sqrt\frac{1+\tau}{1-\tau} \frac{1}{\sqrt{2\pi m}}.
\end{align*}
For $m\geq \frac{1+\tau}{1-\tau}$ we have an upper bound already, but we need one for the remaining $m$. We start with \eqref{eq:lem:TrMnmMono1} and plug in the result from Lemma \ref{lem:asympKernel} for the remaining kernel. This yields
\begin{align*}
\Tr (M_n^{(\tau)})^m \leq 
    \frac{n}{\pi} \frac{1}{\sqrt{m-1}} \int_{\mathbb R_+^2} H(2 \xi_{e_\tau} + n^{-\mu} - \xi_x-\xi_y) 4 e^{-n \frac{m}{m-1} (x^2+y^2)} \cosh\Big(\frac{2nxy}{m-1}\Big) \cosh(n xy) \,dx \,dy. 
\end{align*}
We can write
\begin{align*}
 &\quad  4 e^{-n \frac{m}{m-1} (x^2+y^2)} \cosh\left(\frac{2nxy}{m-1}\right) \cosh(n xy)
 \\
 &   = e^{-n \frac{m}{m-1} (x^2+y^2)} 
    \left(e^{2n\frac{m}{m-1} xy}+e^{-2n\frac{m}{m-1} xy}+e^{2n\frac{m-2}{m-1} xy}+e^{-2n\frac{m-2}{m-1} xy}\right)
    \\
   & = e^{-n\frac{m}{m-1} (x-y)^2}+e^{-n\frac{m}{m-1} (x+y)^2}
    + e^{-\frac{4n}{m} x^2} \left(
    e^{-n\frac{m}{m-1} (y-\frac{m-2}{m}x)^2}+e^{-n\frac{m}{m-1} (y+\frac{m-2}{m}x)^2}\right).
\end{align*}
Note that we may replace the integration domain $\mathbb R_+^2$ by $[0, g_\tau+b n^{-\mu}]^2$ for some constant $b>0$, and
\begin{align*}
    g_\tau = f_\tau \cosh(2 \xi_{e_\tau}) = \frac{1+\tau^2}{2\tau} f_\tau.
\end{align*}
This is allowed because the Heaviside function vanishes outside this region. Now we divide this region into $[f_\tau, g_\tau+b n^{-\mu}]$, and the remaining region. For the latter, we may bound the Heaviside function by $1$ (boundary contribution is of small order), and it follows straightforwardly by steepest descent arguments that the corresponding integral equals
\begin{align*}
    \frac{n}{\pi} \frac{1}{\sqrt{m-1}} \int_0^{f_\tau} \sqrt{\pi \frac{m-1}{mn}} dx = \sqrt\frac{n}{\pi m} f_\tau.
\end{align*}
up to leading order. For the first region we find that the corresponding integral equals
\begin{align*}
     \frac{1}{2} \sqrt\frac{n}{\pi m} \int_{f_\tau}^{g_\tau}  
     \bigg(\erf\Big(\sqrt\frac{m n}{m-1} x\Big)+\erf\Big(\sqrt\frac{m n}{m-1} \left(f_\tau \cosh(2\xi_{e_\tau}-\xi_x)-x\right)\Big)\bigg) \, dx
\end{align*}
to leading order. For $x>1+\tau$, we have
\begin{align*}
    f_\tau \cosh(2\xi_{e_\tau}-\xi_x)-x<f_\tau \cosh(\xi_{e_\tau}) - e_\tau = 0,
\end{align*}
and the two error functions cancel in the limit $n\to\infty$. What remains is
\begin{align*}
    \frac{1}{2} \sqrt\frac{n}{\pi m} \int_{f_\tau}^{g_\tau} 2 \, dx = \sqrt\frac{n}{\pi m} (e_\tau-f_\tau).
\end{align*}
Together with the contribution from $[0,f_\tau]^2$, this gives
\begin{align*}
    \lim_{n\to\infty} \frac{1}{\sqrt{2n}} \Tr (M_n^{(\tau)})^m \leq \frac{1}{\sqrt{2\pi m}} (f_\tau+e_\tau-f_\tau) = \sqrt\frac{1+\tau}{1-\tau} \frac{1}{\sqrt{\pi m}}.
\end{align*}
\end{proof}

\section{Asymptotic analysis at weak non-Hermiticity} \label{Section_asymptotic weak}

In this section, we show Proposition~\ref{prop:asympTracem weak}, i.e. for any fixed $m>0$,
\begin{equation} \label{Tr M power conv weak v2}
\begin{split}
\lim_{n \to \infty} \frac{1}{ 2n } \Tr (M_n^{(\tau)})^m & = \frac{c(\sqrt{m}\,\alpha)}{2}=\frac{e^{-m\alpha^2/2}}{2}  \Big[ I_0\Big(\frac{m\alpha^2}{2}\Big)+I_1\Big(\frac{m\alpha^2}{2}\Big) \Big]. 
\end{split}
\end{equation}
Note that $c(\alpha)$ in \eqref{c(alpha)} can be written as 
\begin{equation} \label{c(alpha) v2}
c(\alpha) = \frac{2}{ \alpha \sqrt{\pi} } \int_0^1 \erf( \alpha \sqrt{1-s^2} )\,ds  =\sum_{k=0}^\infty \frac{(2k-1)!!}{2^k \,k!\,(k+1)!} (-1)^k \alpha^{2k}, 
\end{equation}
see e.g. \cite[Remark 2.7]{byun2021real}. 
Therefore, the right-hand side of \eqref{Tr M power conv weak v2} can be written as
\begin{equation} \label{c m alpha expansion}
\frac{c(\sqrt{m}\,\alpha)}{2} = \sum_{k=0}^\infty \frac{ m^k  }{2^{2k+1}   (k+1)!} \binom{2k}{k}  (-\alpha^2)^{k}.
\end{equation}
Recall that by Lemma~\ref{Lem_Tr M^m sum}, $\Tr (M_n^{(\tau)})^m$ is evaluated as
\begin{align} \label{Tr M^m sum evaluation v2}
 \Tr (M_n^{(\tau)})^m =   \sum_{j_1,\dots,j_m=0}^{n-1}  \Big(\frac{ 1+\tau }{ 2 }\Big)^{2\sum_{k=1}^m j_k+ \frac{m}{2} } \prod_{k=1}^m  \bigg( \sum_{ l=0 }^{n-1}    \frac{ (1-\tau)^{2l }(1+\tau)^{-2l}  \, (2j_k)!  }{ 2^{2l} l! (l+j_{k}-j_{k-1})! (2j_{k-1}-2l)!  } \bigg).
\end{align}
Here and in the sequel, we use the convention $j_0=j_m$.

\begin{rem}
We first discuss the contribution from $l=0$ in the expression \eqref{Tr M^m sum evaluation v2}. 
If $l=0$, the right-hand side of \eqref{Tr M^m sum evaluation v2} is given by 
\begin{align*}
\begin{split}
&\quad  \Big( \frac{1+\tau}{2} \Big)^{m/2}  \sum_{j_1,\dots,j_m=0}^{n-1}  \prod_{k=1}^m  \Big(\frac{ 1+\tau }{ 2 }\Big)^{2j_k}     \frac{   (2j_k)!   }{  (j_{k}-j_{k-1})! (2j_{k-1})!  } 
\\
&=   \Big( \frac{1+\tau}{2} \Big)^{m/2} \sum_{j_1,\dots,j_m=0}^{n-1}  \prod_{k=1}^m  \Big(\frac{ 1+\tau }{ 2 }\Big)^{2j_k}     \frac{ 1  }{  (j_{k}-j_{k-1})!   } 
\\
&=  \Big( \frac{1+\tau}{2} \Big)^{m/2}  \sum_{j=0}^{n-1}  \prod_{k=1}^m  \Big(\frac{ 1+\tau }{ 2 }\Big)^{2j}   = 
\Big( \frac{1+\tau}{2} \Big)^{m/2}
\sum_{j=0}^{n-1}   \Big(\frac{ 1+\tau }{ 2 }\Big)^{2j m}  =  \Big( \frac{1+\tau}{2} \Big)^{m/2} \frac{ 1-\Big(\frac{ 1+\tau }{ 2 }\Big)^{2m n} }{ 1-\Big(\frac{ 1+\tau }{ 2 }\Big)^{2m} }. 
\end{split}
\end{align*}
(From the second to the third line, see Remark \ref{eq:suml=0j1=...=jm}.)
This gives that for $\tau=1-\alpha^2/(2n)$, as $n \to \infty$, 
\begin{align}
\begin{split} \label{l=0 contribution}
\frac{1}{2n}\,\Big( \frac{1+\tau}{2} \Big)^{m/2}  \sum_{j_1,\dots,j_m=0}^{n-1}  \prod_{k=1}^m  \Big(\frac{ 1+\tau }{ 2 }\Big)^{2j_k}     \frac{   (2j_k)!   }{  (j_{k}-j_{k-1})! (2j_{k-1})!  } 
=  \frac{1-e^{ -\frac{\alpha^2m}{2} }}{\alpha^2 m} +\mathcal O\Big(\frac1n\Big). 
\end{split}
\end{align}
\end{rem}

\subsection{Expectation and variance revisited}

It is instructive to first consider the cases $m=1,2$ before dealing with the general $m$.  
By Proposition~\ref{Prop_generating matrix implementation} (ii), the analysis for $m=1,2$ below provides an alternative and more unified proof of  \cite[Theorems 2.1 and 2.3]{byun2021real}. 

\subsubsection{The case $m=1$}
We first consider the simplest case $m=1.$
Then by \eqref{Tr M^m sum evaluation v2}, we have
\begin{align}
\begin{split}
 \Tr (M_n^{(\tau)}) =    \sum_{ l=0 }^{n-1}  \frac{1}{(l!)^2}  \frac{ (1-\tau)^{2l }(1+\tau)^{-2l}  }{ 2^{2l}   } \sum_{j=0}^{n-1}  \Big(\frac{ 1+\tau }{ 2 }\Big)^{2j+\frac12 }   \frac{  (2j)!  }{ (2j-2l)!  } .
\end{split}
\end{align}
Note that for $\tau=1-\alpha^2/(2n)$, we have 
\begin{align*}
\frac{1}{2n}  \frac{ (1-\tau)^{2l }(1+\tau)^{-2l}  }{ 2^{2l}    } \sum_{j=0}^{n-1}  \Big(\frac{ 1+\tau }{ 2 }\Big)^{2j+\frac12 }   \frac{  (2j)!  }{ (2j-2l)!  } \sim     \frac{ (\alpha/2)^{4l}  }{(2n)^{2l+1}} \sum_{j=0}^{n-1}  \Big(\frac{ 1+\tau }{ 2 }\Big)^{2j+\frac12}   \frac{  (2j)!  }{ (2j-2l)!  }. 
\end{align*}
For $l=o(n)$, the Riemann sum approximation gives 
\begin{align*}
 \frac{1}{(2n)^{2l+1}} \sum_{j=0}^{n-1}  \Big(\frac{ 1+\tau }{ 2 }\Big)^{2j+\frac12}   \frac{  (2j)!  }{ (2j-2l)!  } & = \frac{1}{2n} \sum_{j=0}^{n-1}  \Big( 1- \frac{\alpha^2}{4n} \Big)^{2j}  \frac{2j}{2n} \frac{2j-1}{2n} \dots \frac{2j-2l+1}{2n}
\\
& \sim\frac12 \int_0^1 e^{ -\frac{\alpha^2 x}{2} } x^{2l}\,dx=  \frac{ 2^{2l} }{ \alpha^{4l+2} } \gamma \Big(2l+1, \frac{\alpha^2}{2} \Big). 
\end{align*}
Note here that for $l=0$, it matches with \eqref{l=0 contribution} since $\gamma(1,x)=1-e^{-x}$. 
Combining the above, for $l=o(n)$, 
\begin{align}
\frac{1}{2n}\frac{ (1-\tau)^{2l }(1+\tau)^{-2l}  }{ 2^{2l} (l!)^2   }  \sum_{j=0}^{n-1}  \Big(\frac{ 1+\tau }{ 2 }\Big)^{2j+\frac12}   \frac{  (2j)!  }{ (2j-2l)!  }  \sim  \frac{ 1 }{ \alpha^{2} } \frac{1}{ 2^{2l}(l!)^2}   \gamma \Big(2l+1, \frac{\alpha^2}{2} \Big). 
\end{align}
This gives that for a sufficiently large $L>0$, 
\begin{align*}
\frac{1}{2n} \Tr (M_n^{(\tau)})  &\sim  \frac{ 1 }{ \alpha^{2} }  \sum_{l=0}^L \frac{1}{ 2^{2l}(l!)^2}   \gamma \Big(2l+1, \frac{\alpha^2}{2} \Big)  + \sum_{l=L+1}^{n-1}   \frac{1}{(l!)^2}   \frac{ (\alpha/2)^{4l}  }{(2n)^{2l+1}} \sum_{j=0}^{n-1}  \Big(\frac{ 1+\tau }{ 2 }\Big)^{2j+\frac12}   \frac{  (2j)!  }{ (2j-2l)!  } . 
\end{align*}
Note here that for any $l=0,\dots, n-1,$
\begin{align}
 \frac{ 1  }{(2n)^{2l+1}} \sum_{j=0}^{n-1}  \Big(\frac{ 1+\tau }{ 2 }\Big)^{2j+\frac12}   \frac{  (2j)!  }{ (2j-2l)!  }  &= \frac{1}{2n}  \sum_{j=0}^{n-1}  \Big(\frac{ 1+\tau }{ 2 }\Big)^{2j+\frac12}   \frac{2j}{2n} \frac{2j-1}{2n} \dots \frac{2j-2l+1}{2n} < \frac12. 
\end{align}
Thus as $L \to \infty$ keeping $L=o(N)$, we have 
$$
\sum_{l=L+1}^{n-1}   \frac{1}{(l!)^2}   \frac{ (\alpha/2)^{4l}  }{(2n)^{2l+1}} \sum_{j=0}^{n-1}  \Big(\frac{ 1+\tau }{ 2 }\Big)^{2j+\frac12}   \frac{  (2j)!  }{ (2j-2l)!  } \le \sum_{l=L+1}^{n-1} \frac{ (\alpha/2)^{4l} }{2 (l!)^2} \to 0. 
$$
Therefore we obtain 
\begin{equation}
\lim_{n \to \infty} \frac{1}{2n} \Tr (M_n^{(\tau)}) =  \frac{ 1 }{ \alpha^{2} }  \sum_{l=0}^\infty \frac{1}{ 2^{2l}(l!)^2}   \gamma \Big(2l+1, \frac{\alpha^2}{2} \Big).
\end{equation}
Now Proposition~\ref{prop:asympTracem weak} for $m=1$ follows from the following lemma. 

\begin{lem}
We have 
\begin{equation}
  \frac{ 1 }{ \alpha^{2} }  \sum_{l=0}^\infty \frac{1}{ 2^{2l}(l!)^2}   \gamma \Big(2l+1, \frac{\alpha^2}{2} \Big) = \frac{c(\alpha)}{2} .
\end{equation}
\end{lem}
\begin{proof}
Note that 
\begin{equation} \label{gamma series}
\gamma(j,x)= \sum_{s=0}^\infty \frac{(-1)^s x^{j+s}}{ s!(j+s) },
\end{equation}
see \cite[Eq.(8.7.1)]{olver2010nist}. 
Thus we have 
\begin{equation}
\begin{split}
 \frac{ 1 }{ \alpha^{2} }  \frac{1}{ 2^{2l}(l!)^2}  \gamma \Big(2l+1, \frac{\alpha^2}{2} \Big) 
 &=   \frac{ 1 }{ 2^{4l+1}(l!)^2} \sum_{s=0}^\infty \frac{1}{ 2^s s!(2l+1+s) } (-\alpha^2)^{s+2l} . 
\end{split}
\end{equation}
This gives that
\begin{align*}
 \frac{ 1 }{ \alpha^{2} }  \sum_{l=0}^\infty \frac{1}{ 2^{2l}(l!)^2}   \gamma \Big(2l+1, \frac{\alpha^2}{2} \Big)  &= \sum_{k=0}^\infty \bigg( \sum_{\substack{s,l=0 \\ s+2l=k }}^\infty   \frac{ 1 }{ 2^{4l+s+1}(l!)^2  s!(2l+1+s)}  \bigg) (-\alpha^2)^{k}
 \\
  &= \sum_{k=0}^\infty \frac{1}{k+1} \bigg( \sum_{\substack{s,l=0 \\ s+2l=k }}^\infty   \frac{ 2^{s-2k-1} }{ (l!)^2  s! }  \bigg) (-\alpha^2)^{k} . 
\end{align*}
Then by \eqref{c m alpha expansion}, it suffices to show that
\begin{equation}
 \sum_{\substack{s,l=0 \\ s+2l=k }}^\infty 2^{s}  \frac{k!}{(l!)^2 s!}  =  \binom{2k}{k}. 
\end{equation}
This combinatorial identity follows by comparing the coefficient of the $(xy)^k$ term in 
\begin{equation}
(x^2+2xy+y^2)^k = (x+y)^{2k},
\end{equation}
which completes the proof. 
\end{proof}

\subsubsection{The case $m=2$}

For $m=2$, by \eqref{Tr M^m sum evaluation v2},  we have 
\begin{align}
\begin{split}
 \Tr (M_n^{(\tau)})^2 
&=  \sum_{j_1,j_2=0}^{n-1}  \Big(\frac{ 1+\tau }{ 2 }\Big)^{2(j_1+j_2)+1} 
\\
&\quad \times  \sum_{ l_1,l_2=0 }^{n-1}    \frac{ (1-\tau)^{2(l_1+l_2) }(1+\tau)^{-2(l_1+l_2) }  \, (2j_1)! (2j_2)!  }{ 2^{2(l_1+l_2) } l_1! l_2! (l_1+j_{1}-j_{2})! (l_2+j_{2}-j_{1})! (2j_{2}-2l_1)! (2j_{1}-2l_2)!   }   .
\end{split}
\end{align}
Here, the summand does not vanish only for the set of indices 
\begin{equation}
j_1-j_2 \ge -l_1, \qquad j_2-j_1 \ge -l_2, \qquad j_2 \ge l_1, \qquad j_1 \ge l_2. 
\end{equation}

As before, by the rapid decay of the factor $1/(l_1! l_2!)$, it suffices to consider the case $l_1$ and $l_2$ are finite. 
Furthermore, due to the term 
$$
 \frac{1 }{  (l_1+j_{1}-j_{2})! (l_2+j_{2}-j_{1})!  },
$$
it is enough to consider the case $j_1-j_2$ is finite. 
By letting $M=j_1-j_2$, it follows that 
\begin{align}
\begin{split}
\frac{1}{2n}  \Tr (M_n^{(\tau)})^2  & \sim  \sum_{M=-\infty}^\infty   \sum_{ l_1,l_2=0 }^{\infty}  \frac{1 }{ l_1! l_2! (l_1+M)! (l_2-M)! } 
\\
&\quad \times \frac{1}{2n} \sum_{j=0}^{n-1} \Big(\frac{ 1+\tau }{ 2 }\Big)^{4j}      \frac{ (1-\tau)^{2(l_1+l_2) }(1+\tau)^{-2(l_1+l_2) }  \, (2j+M)! (2j-M)!  }{ 2^{2(l_1+l_2) }  (2j-2M-2l_1)! (2j-2l_2)!   }.
\end{split}
\end{align}
Here, by the Riemann sum approximation, we have 
\begin{align}
\begin{split}
&\quad \frac{1}{2n} \sum_{j=0}^{n-1} \Big(\frac{ 1+\tau }{ 2 }\Big)^{4j}      \frac{ (1-\tau)^{2(l_1+l_2) }(1+\tau)^{-2(l_1+l_2) }  \, (2j+M)! (2j-M)!  }{ 2^{2(l_1+l_2) }  (2j-2M-2l_1)! (2j-2l_2)!   }
\\
& \sim \frac{(\alpha/2)^{4l}}{2}  \int_0^1 e^{ - \alpha^2 x   } x^{2l} \,dx = \frac{1}{ \alpha^2 \,2^{4l+1} }\gamma \Big(2l+1, \alpha^2 \Big),
\end{split}
\end{align}
where we write $l=l_1+l_2$. 
Therefore we obtain 
\begin{equation}
\lim_{n \to \infty} \frac{1}{2n}  \Tr (M_n^{(\tau)})^2=  \frac{1}{\alpha^2} \sum_{M=-\infty}^{ \infty } \sum_{ \substack{l_1,l_2=0 \\ l_1+l_2=l } }^\infty   \frac{1}{ 2^{4l+1} }  \frac{1}{  l_1! l_2! (l_1+M)!(l_2-M)! } \gamma \Big(2l+1, \alpha^2 \Big).
\end{equation}
Now it suffices to show the following.

\begin{lem}
We have 
\begin{equation}
\frac{1}{\alpha^2} \sum_{M=-\infty}^\infty \sum_{ \substack{l_1,l_2=0 \\ l_1+l_2=l } }^\infty   \frac{1}{ 2^{4l+1} }  \frac{1}{  l_1! l_2!  (l_1+M)! (l_2-M)! } \gamma \Big(2l+1, \alpha^2 \Big) = \frac{ c(\sqrt{2}\alpha) }{2}. 
\end{equation}
\end{lem}

\begin{proof}
Using \eqref{gamma series}, 
we have 
\begin{align*}
&\quad \frac{1}{\alpha^2} \sum_{ \substack{l_1,l_2=0 \\ l_1+l_2=l } }^\infty   \frac{1}{ 2^{4l+1} }  \frac{1}{  (l_1!)^2 (l_2!)^2 } \gamma \Big(2l+1, \alpha^2 \Big) 
=\sum_{ \substack{l_1,l_2=0 \\ l_1+l_2=l } }^\infty  \frac{1}{  2^{4l+1} }  \frac{1}{  (l_1)!^2 (l_2)!^2 } \sum_{s=0}^\infty \frac{ (-\alpha^2)^{s+2l} }{ s! (2l+1+s) }
\\
&= \sum_{k=0}^\infty \bigg(   \sum_{ \substack{s,l=0 \\ s+2l=k } }^\infty  \sum_{ \substack{l_1,l_2=0 \\ l_1+l_2=l } }^\infty   \frac{1}{ 2^{4l+1} } \frac{1}{  (l_1)!^2 (l_2)!^2 } \frac{ 1 }{ s! (2l+1+s) } \bigg) (-\alpha^2)^k 
\\
&= \sum_{k=0}^\infty \frac{1}{k+1}\bigg(     \sum_{ \substack{s,l=0 \\ s+2l=k } }^\infty  \sum_{ \substack{l_1,l_2=0 \\ l_1+l_2=l } }^\infty   \frac{1}{ 2^{4l+1} } \frac{1}{  (l_1)!^2 (l_2)!^2 } \frac{ 1 }{ s!  } \bigg) (-\alpha^2)^k . 
\end{align*}
Then by \eqref{c m alpha expansion}, it suffices to show that
\begin{equation} \label{combinatorial identity m2}
\sum_{M=-\infty}^\infty \sum_{ \substack{s,l=0 \\ s+2l=k } }^\infty  \sum_{ \substack{l_1,l_2=0 \\ l_1+l_2=l } }^\infty   2^{s-2l} \frac{k!}{  l_1! l_2! (l_1+M)! (l_2-M)! s! } =  \binom{2k}{k}. 
\end{equation}
Note that for any $k$, there are only finite numbers of non-trivial summands in the left-hand side of this identity. 
To prove this combinatorial identity, we can compare the coefficient of the $(xy)^k$ term in the expansion of 
\begin{equation} \label{binomial m2}
\Big(\frac{x^2}{2}+\frac{x^2}{2}+2xy+\frac{y^2}{2}+\frac{y^2}{2} \Big)^k = (x+y)^{2k}.
\end{equation}
To be more precise, the left-hand side of \eqref{combinatorial identity m2} can be obtained as the coefficient of the $(xy)^k$ term in the expansion of the left-hand side of \eqref{binomial m2}. 
For this, we choose $l_1$ and $l_2$ instances of the first two terms $x^2/2$ and $x^2/2$, respectively, $s$ instances of the term $2xy$, and $l_1+M$ and $l_2-M$ instances of the last two terms $y^2/2$ and $y^2/2$, respectively. We then collect all possible combinations of $l_1$, $l_2$, and $M$ such that 
$$
l_1+l_2+s+(l_1+M)+(l_2-M)= s+2(l_1+l_2) = s +2l = k. 
$$
This completes the proof. 
\end{proof}

\subsection{Proof of Proposition~\ref{prop:asympTracem weak}}


We now consider the general case with $m \ge 1$.
Let us first rewrite \eqref{Tr M^m sum evaluation v2} as 
\begin{align*}
 \Tr (M_n^{(\tau)})^m &=  \sum_{j_1,\dots,j_m=0}^{n-1}  \Big(\frac{ 1+\tau }{ 2 }\Big)^{2\sum_{k=1}^m j_k+ \frac{m}{2} }\sum_{ l_1,\dots, l_m=0 }^{n-1}   \frac{ (1-\tau)^{2(l_1+\dots+l_m) }(1+\tau)^{-2(l_1+\dots +l_m)}   }{ 2^{2(l_1+\dots+l_m )} l_1! \dots l_m!  }
 \\
 & \quad \times  \frac{ (2j_1)! \dots  (2j_m)!  }{  (l_1+M_1)! \dots (l_m+M_m)! \, (2j_{m}-2l_1)! (2j_{1}-2l_2)! \dots (2j_{m-1}-2l_m)!  },
\end{align*}
where $M_k=j_k-j_{k-1}$ ($k=1,\dots, m$). 
Again, it suffices to consider the case that the $l_j$'s are finite. 
As a consequence, due to the terms 
$$
 \frac{ 1  }{  (l_1+M_{1})! \cdots (l_m+M_m)!  } 
$$
in \eqref{Tr M^m sum evaluation v2},  it is enough to take the case $j_k-j_{k-1}$ ($k=1,\dots , m$) finite into account. 
We write 
$$
\wt{M}_k= \sum_{p=2}^k M_p= j_k-j_1.  
$$
Then by combining the above, we obtain 
\begin{align*}
\begin{split}
&\quad \frac{1}{2n}  \Tr (M_n^{(\tau)})^m   \sim   \sum_{ \substack{ M_1,\dots , M_m=-\infty \\ M_1+\dots+M_m=0}  }^\infty 
 \sum_{ \substack{l_1,\dots l_m=0 \\ l_1+\dots+l_m=l } }^\infty    \frac{1 }{ l_1! \dots  l_m! (l_1+M_1)! \dots (l_m+M_m)! } 
\\
&\quad \times \frac{1}{2n} \sum_{j=0}^{n-1} \Big(\frac{ 1+\tau }{ 2 }\Big)^{2mj+\frac{m}{2}}   \frac{(1-\tau)^{2l }  }{ (2(1+\tau))^{2l }   }   \frac{  (2j)! (2j+2\wt{M}_2 ) \dots (2j+2\wt{M}_m )!  }{  (2j+ 2\wt{M}_m -2l_1)! (2j-2l_2)!\dots (2j+ 2 \wt{M}_{m-1} -2l_m)!   },
\end{split}
\end{align*}
where $l=l_1+\dots+l_m$. 
Note that 
\begin{align*}
&\quad (1-\tau)^{2l}  \frac{  (2j)! (2j+2\wt{M}_2 ) \dots (2j+2\wt{M}_m )!  }{  (2j+ 2\wt{M}_m -2l_1)! (2j-2l_2)!\dots (2j+ 2 \wt{M}_{m-1} -2l_m)!   } 
\\
&= \alpha^{4l} \Big(\frac{ 2j+2 \wt{M}_m-2l_1+1 }{ 2n } \dots \frac{2j}{2n}\Big) \dots \Big( \frac{ 2j+2 \wt{M}_{m-1}-2l_m+1 }{ 2n } \dots \frac{2j+2 \wt{M}_m }{2n} \Big). 
\end{align*}
Then it follows from the Riemann sum approximation that 
\begin{align}
\begin{split}
&\quad \frac{1}{2n} \sum_{j=0}^{n-1} \Big(\frac{ 1+\tau }{ 2 }\Big)^{2mj+\frac{m}{2}}     \frac{(1-\tau)^{2l }  }{ (2(1+\tau))^{2l }   }   \frac{  (2j)! (2j+2\wt{M}_2 ) \dots (2j+2\wt{M}_m )!  }{  (2j+ 2\wt{M}_m -2l_1)! (2j-2l_2)!\dots (2j+ 2 \wt{M}_{m-1} -2l_m)!   }
\\
& \sim \frac{ (\alpha/2)^{4l} }{ 2 }   \int_0^1 e^{ -\frac{ m\alpha^2 }{2} x }  x^{2l} \,dx =  \frac{1}{\alpha^2}  \frac{1}{ 2^{2l} m^{2l+1} } \gamma\Big( 2l+1, \frac{\alpha^2m}{2} \Big) .
\end{split}
\end{align}
Therefore we obtain 
\begin{equation}
\begin{split}
&\quad \lim_{n \to \infty} \frac{1}{2n}  \Tr (M_n^{(\tau)})^m
 \\
 & =  \frac{1}{\alpha^2} \sum_{ \substack{ M_1,\dots , M_m=-\infty \\ M_1+\dots+M_m=0}  }^\infty 
 \sum_{ \substack{l_1,\dots l_m=0 \\ l_1+\dots+l_m=l } }^\infty   \frac{1 }{ l_1! \dots  l_m! (l_1+M_1)! \dots (l_m+M_m)! }   \frac{1}{ 2^{2l} m^{2l+1} } \gamma\Big( 2l+1, \frac{\alpha^2m}{2} \Big) .
\end{split}
\end{equation}
Then the following lemma completes the proof of Proposition~\ref{prop:asympTracem weak}.

\begin{lem}
We have 
\begin{equation*}
 \frac{1}{\alpha^2}  \sum_{ \substack{ M_1,\dots , M_m=-\infty \\ M_1+\dots+M_m=0}  }^\infty  
 \sum_{ \substack{l_1,\dots l_m=0 \\ l_1+\dots+l_m=l } }^\infty   \frac{1 }{ l_1! \dots  l_m! (l_1+M_1)! \dots (l_m+M_m)! }   \frac{1}{ 2^{2l} m^{2l+1} } \gamma\Big( 2l+1, \frac{\alpha^2m}{2} \Big)  = \frac{ c(\sqrt{m}\alpha) }{2}.
\end{equation*}
\end{lem}
\begin{proof}
Using \eqref{gamma series}, we have 
\begin{align*}
&\quad \frac{1}{\alpha^2} \sum_{ \substack{l_1,\dots l_m=0 \\ l_1+\dots+l_m=l } }^\infty     \frac{1 }{ l_1! \dots  l_m! (l_1+M_1)! \dots (l_m+M_m)! }   \frac{1}{ 2^{2l} m^{2l+1} } \gamma\Big( 2l+1, \frac{\alpha^2m}{2} \Big) 
\\
&= \frac{1}{\alpha^2} \sum_{ \substack{l_1,\dots l_m=0 \\ l_1+\dots+l_m=l } }^\infty    \frac{1 }{ l_1! \dots  l_m! (l_1+M_1)! \dots (l_m+M_m)! }   \frac{1}{ 2^{2l} m^{2l+1} } \sum_{ s=0 }^\infty \frac{ (-1)^{s } }{ s! (s+2l+1) } \Big( \frac{\alpha^2m}{2} \Big)^{s+2l+1} 
\\
&=  \sum_{ \substack{l_1,\dots l_m=0 \\ l_1+\dots+l_m=l } }^\infty   \frac{1}{ 2^{4l+1}  }   \frac{1 }{ l_1! \dots  l_m! (l_1+M_1)! \dots (l_m+M_m)! } \sum_{ s=0 }^\infty \frac{ m^s }{2^s s! (s+2l+1) } (-\alpha^2)^{s+2l} .
\end{align*}
This can be rewritten as 
\begin{align*}
& \quad  \sum_{ \substack{l_1,\dots l_m=0 \\ l_1+\dots+l_m=l } }^\infty  \frac{1}{ 2^{4l+1}  }     \frac{1 }{ l_1! \dots  l_m! (l_1+M_1)! \dots (l_m+M_m)! } \sum_{ s=0 }^\infty \frac{ m^s }{2^s s! (s+2l+1) } (-\alpha^2)^{s+2l} 
\\
&= \sum_{k=0}^\infty \sum_{ \substack{s,l=0\\ s+2l=k } }^\infty     \sum_{ \substack{l_1,\dots l_m=0 \\ l_1+\dots+l_m=l } }^\infty  \frac{1}{ 2^{4l+1}  }    \frac{1 }{ l_1! \dots  l_m! (l_1+M_1)! \dots (l_m+M_m)! }      \frac{ m^s }{2^s s! (s+2l+1) }              (-\alpha^2)^{k} .
\end{align*}
By \eqref{c m alpha expansion}, all we need to show is
\begin{equation} \label{combinatorial identity m}
 \sum_{ \substack{ M_1,\dots , M_m=-\infty \\ M_1+\dots+M_m=0}  }^\infty   \sum_{ \substack{s,l=0\\ s+2l=k } }^\infty     \sum_{ \substack{l_1,\dots l_m=0 \\ l_1+\dots+l_m=l } }^\infty   \frac{ 2^s\,m^{s-k}\, k! }{ l_1! \dots  l_m! (l_1+M_1)! \dots (l_m+M_m)! \,s! } =  \binom{2k}{k}.
\end{equation}
As before, this identity follows by comparing the coefficient of $(xy)^k$ term in 
\begin{equation} \label{binomial m}
\Big(\frac{x^2}{m}+\dots+\frac{x^2}{m}+2xy+\frac{y^2}{m}+\dots+\frac{y^2}{m} \Big)^k = (x+y)^{2k}.
\end{equation}
Namely, we choose $l_1,l_2,\dots,l_m$ instances of the first $m$ terms $x^2/m, \dots, x^2/m$, respectively, $s$ instances of the term $2xy$, and $l_1+M_1, \dots, l_m+M_m$ instances of the last $m$ terms $y^2/m, \dots, y^2/m$, respectively. 
We then collect all possible combinations of $l_1,\dots, l_m$, and $M_1,\dots,M_m$ such that $M_1+\dots M_m=0$ and 
$$
\Big(l_1+ \dots +l_m\Big)+ s+\Big( (l_1+M_1)+\dots (l_m+M_m) \Big)= s+2(l_1+\dots+l_m) = s +2l = k. 
$$
Notice here that the exponent of $m$ is 
$$
-\Big(l_1+ \dots +l_m\Big)- \Big( (l_1+M_1)+\dots (l_m+M_m) \Big) = -2l = s-k. 
$$
This completes the proof. 
\end{proof}

\appendix

\section{Auxiliary lemmas} \label{Appendix A_auxiliary lemmas}

\begin{lem} \label{Lem_d(alpha) integral rep}
For fixed $\alpha>0$, we have
\begin{equation}
- \sum_{m=1}^\infty \frac{ c(\sqrt{m}\,\alpha) }{2m} = \frac{2}{\pi} \int_0^1 \log\Big(1-e^{-\alpha^2 s^2}\Big) \sqrt{1-s^2} \, ds.
\end{equation}
\end{lem}

\begin{proof}
We notice that
\begin{align*}
&\quad \frac{1}{\alpha} \int_0^1 \erf(\alpha \sqrt{1-s^2}) \, ds = \frac{1}{\alpha} \int_0^1 \erf(\alpha s) \frac{s \, ds}{\sqrt{1-s^2}}\\
&= \Big[- \frac{1}{\alpha} \erf(\alpha s) \sqrt{1-s^2}\Big]_0^1
+ \frac{2}{\sqrt\pi}\int_0^1 e^{-(\alpha s)^2} \sqrt{1-s^2} \, ds
= \frac{2}{\sqrt\pi} \int_0^1 e^{-\alpha^2 s^2} \sqrt{1-s^2} \, ds.
\end{align*}
We conclude, using \eqref{c(alpha)} and \eqref{c(alpha) v2}, that
\begin{align*}
- \sum_{m=1}^\infty \frac{ c(\sqrt{m}\,\alpha) }{2m} = - \frac{2}{\pi} \sum_{m=1}^\infty \int_0^1 \frac{1}{m} e^{-m \alpha^2 s^2} \sqrt{1-s^2} \, ds
= \frac{2}{\pi} \int_0^1 \log\Big(1-e^{-\alpha^2 s^2}\Big) \sqrt{1-s^2} \, ds.
\end{align*}
We need to prove that we may indeed interchange the order of summation and integration. We start with the observation, based on steepest descent arguments, that there exists an integer $M>0$ such that $m>M$ implies that
\begin{align*}
\Big|\int_0^1 e^{-m\alpha^2 s^2}\sqrt{1-s^2} ds- \frac{1}{2\alpha}\sqrt{\frac{\pi}{m}}\Big|\leq C \frac{1}{m\sqrt m}
\end{align*}
for some uniform constant $C$ (depending only on $\alpha$). Now let $\varepsilon>0$ satisfy $\varepsilon<M^{-2}$. Let $M_\varepsilon$ be the integer in $(\varepsilon^{-1/2},1+\varepsilon^{-1/2}]$. Then we have $M_\varepsilon>M$, and thus
\begin{align*}
\sum_{m=1}^\infty \frac{1}{m}\int_0^\varepsilon e^{-m\alpha^2 s^2} \sqrt{1-s^2} ds
&=\sum_{m=1}^{M_\varepsilon} \frac{1}{m} \int_0^\varepsilon e^{-m\alpha^2 s^2} \sqrt{1-s^2} ds+\sum_{m=M_\varepsilon+1}^{\infty} \frac{1}{m} \int_0^\varepsilon e^{-m\alpha^2 s^2} \sqrt{1-s^2} ds\\
&\leq \sum_{m=1}^{M_\varepsilon} \frac{\varepsilon}{m} 
+ \sum_{m=M_\varepsilon+1}^\infty \frac{1}{m}\Big(\frac{1}{2\alpha}\sqrt{\frac{\pi}{m}}+\frac{C}{m\sqrt m}\Big) \\
&\leq \varepsilon (1+\log M_\varepsilon)
+ \frac{\sqrt\pi}{\alpha} \frac{1}{ \sqrt{M_\varepsilon}}+\frac{2}{3}
\frac{C}{M_\varepsilon\sqrt{M_\varepsilon}}.
\end{align*}
This tends to $0$ as $\varepsilon\to 0$, and we are done.
\end{proof}

\begin{lem} \label{lem:eCoshProdInt2}
Let $k>1$ be an integer and let $x_0, x_k\in\mathbb R$. Then we have
\begin{align*}
    \int_{\mathbb R_+^{k-1}} e^{-\sum_{j=1}^{k-1} x_j^2} \prod_{j=1}^k \cosh(x_{j-1} x_{j}) \, dx_1\cdots dx_{k-1}
    = \frac{1}{\sqrt{k}} \left(\frac{\pi}{2}\right)^\frac{k-1}{2} e^{\frac{k-1}{k}\frac{x_0^2+x_{k}^2}{2}}
    \cosh \Big(\frac{x_0 x_k}{k} \Big).
\end{align*}
\end{lem}

\begin{proof}
First, we verify that the statement is true for $k=2$.
\begin{align*}
    &\quad \int_{\mathbb R_+} e^{-x_1^2} \cosh(x_0 x_1) \cosh(x_1 x_2) \,dx_1
    \\
    &= \frac{1}{4} \int_{\mathbb R_+} e^{-x_1^2}
    \left(e^{x_1(x_0+x_2)}+e^{-x_1(x_0+x_2)}+e^{x_1(x_0-x_2)}+e^{-x_1(x_0-x_2)}\right) \, dx_1\\
    &= \frac{1}{4} \int_{\mathbb R} e^{-x_1^2}
    \left(e^{x_1(x_0+x_2)}+e^{-x_1(x_0+x_2)}\right) \, dx_1
    \\
    &= \frac{\sqrt\pi}{4} e^{\frac{(x_0+x_2)^2}{4}}
    +\frac{\sqrt\pi}{4} e^{\frac{(x_0-x_2)^2}{4}}= \frac{\sqrt\pi}{2} e^{\frac{x_0^2+x_2^2}{4}} \cosh \left(\frac{x_0 x_2}{2}\right).
\end{align*}
Now we use the induction argument and suppose that the statement is true for $k$. 
Then we have
\begin{align*}
 &\quad \int_{\mathbb R_+^{k}} e^{-\sum_{j=1}^{k} x_j^2} \prod_{j=1}^{k+1} \cosh(x_{j-1} x_{j}) \, dx_1\cdots dx_{k} 
 \\
 &= \frac{1}{\sqrt{k}} \left(\frac{\pi}{2}\right)^\frac{k-1}{2} e^{\frac{k-1}{k} \frac{x_0^2}{2}} \int_{\mathbb R_+} e^{-\frac{k+1}{2 k} x_k^2} \cosh \Big(\frac{x_0 x_k}{k}\Big) \cosh(x_k x_{k+1}) \, dx_k
 \\
 &= \frac{1}{\sqrt{k}} \left(\frac{\pi}{2}\right)^\frac{k-1}{2} e^{\frac{k-1}{k} \frac{x_0^2}{2}}
    \sqrt{\frac{2k}{k+1}}\int_{\mathbb R_+} e^{-x_k^2} \cosh\Big( \frac{1 }{k} \sqrt{\frac{2k}{k+1}} x_0 x_k\Big) \cosh\Big( \sqrt{\frac{2k}{k+1}}x_k x_{k+1}\Big) \, dx_k
    \\
   &=  \frac{\sqrt 2}{\sqrt{k+1}} 
    \left(\frac{\pi}{2}\right)^\frac{k-1}{2} e^{\frac{k-1}{k} \frac{x_0^2}{2}}   \frac{\sqrt\pi}{2} e^{\frac{1}{4}\left(\frac{2 }{k(k+1)}x_0^2+\frac{2k}{k+1}x_{k+1}^2\right)} \cosh\Big(\frac{x_0 x_{k+1}}{k+1}\Big), 
\end{align*}
and the result follows after noting that $\frac{k-1}{k}+\frac{1}{k(k+1)}
    = \frac{k}{k+1}$.
\end{proof}

\section{An equivalent determinantal formula}  \label{Appendix B_equivalent det formula}

The following proposition is given in \cite[Section 3]{MR2430570}.

\begin{prop}[Cf. Section 3 in \cite{MR2430570}] \label{Prop_FN finite N}
We have 
\begin{equation}
	p_{N,k}^{(\tau)} =  \Big( \frac{1+\tau}{2} \Big)^{ N(N-1)/4 } \frac{1}{ 2^{N/2} \prod_{l=1}^{N} \Gamma(l/2) } [z^{k/2}] \det\Big[ z \, \alpha_{2j-1,2m} + \beta_{2j-1,2m}[1] \Big]_{j,m=1}^{N/2}, 
\end{equation}
where 
\begin{align}
\alpha_{2j-1,2m}[1] & =2^m(m-1)! \sum_{p=1}^{m} \frac{ \Gamma(j+p-3/2) }{ 2^{p-1} (p-1)! },
\\
\label{beta FN}
\beta_{2j-1,2m}[1] &= -4 \im \int_{ \R } \,dx \int_{0}^\infty \,dy\,  e^{y^2-x^2} \erfc\Big( \sqrt{ \frac{2}{1-\tau} } y \Big) (x+iy)^{2j-2}(x-iy)^{2m-1} .
\end{align}
In particular, we have 
\begin{equation} \label{p2n0 FN}
	p_{2n,0}^{(\tau)} =  \Big( \frac{1+\tau}{2} \Big)^{ n(2n-1)/2 } \frac{1}{ 2^{n} \prod_{l=1}^{2n} \Gamma(l/2) } \det\Big[  \beta_{2j-1,2m}[1] \Big]_{j,m=1}^{n}. 
\end{equation}
\end{prop}

\begin{ex}
By \eqref{Mn tau jk}, we have
\begin{align}
M_2^{(\tau)}= \begin{bmatrix}
\frac{\sqrt{2(1+\tau)}}{2} & \frac{ (1-\tau)\sqrt{1+\tau} }{4}
\\
\frac{ (1-\tau)\sqrt{1+\tau} }{4} & \frac{ \sqrt{2(1+\tau)} (3+2t+3t^2) }{16}
\end{bmatrix}
\end{align}
Then one can check that the formulas \eqref{p 2n 0 Mn tau} for $n=1,2$ give rise to 
\begin{equation} \label{p2 p4 tau}
p_{2,0}^{(\tau)} =  1-\frac{\sqrt{2(1+\tau)} }{2}, \qquad p_{4,0}^{(\tau)}  =   \frac{9+3\tau+3\tau^2+\tau^3}{8}- \frac{ \sqrt{2(1+\tau)} (11+2\tau+3\tau^2) }{16}.
\end{equation}
In particular, for $\tau=0$, we have 
\begin{equation} \label{p2 p4 tau0}
p_{2,0}^{(0)}= 1-\frac{\sqrt{2}}{2}, \qquad 
p_{4,0}^{(0)}= \frac{18-11\sqrt{2}}{ 16 }.
\end{equation}
The formula \eqref{p2 p4 tau0} also appeared in \cite{edelman1997probability}.
It is obvious that in the symmetric case when $\tau=1$, $p_{2,0}^{(1)}=p_{4,0}^{(1)}=0.$

On the other hand, by direct computations using \eqref{beta FN}, we have
\begin{align*}
&\beta_{1,2}[1]
= 2\sqrt{\pi} \frac{ \sqrt{2}-s }{ s }, \qquad \beta_{3,4}[1]= \sqrt{\pi} \frac{  12\sqrt{2}-16\sqrt{2}s+12\sqrt{2}s^4-7s^5  }{  2s^5 } , 
\\
&\beta_{3,2}[1]= -\sqrt{\pi} \frac{ 2\sqrt{2}- 2\sqrt{2}s^2+s^3  }{  s^3 }
,\qquad \beta_{1,4}[1]=  -\sqrt{\pi}  \frac{  2\sqrt{2}-6\sqrt{2} s^2+5s^3   }{ s^3 },
\end{align*}
where $s=\sqrt{1+\tau}$. 
Then by \eqref{p2n0 FN}, we have \eqref{p2 p4 tau}. 
\end{ex}

\bibliographystyle{abbrv}
\bibliography{RMTbib}
\end{document}